\documentclass[12pt,draftcls,onecolumn]{IEEEtran}

\usepackage{graphicx}
\usepackage{array}
\usepackage{mathrsfs}
\usepackage{amsfonts}
\usepackage{setspace}
\usepackage{amsmath}
\usepackage{amssymb}
\usepackage{mathrsfs}
\input xy
\xyoption{all}

\newtheorem{theorem}{Theorem}
\newtheorem{lemma}{Lemma}

\newtheorem{definition}{Definition}
\newtheorem{corollary}{Corollary}

\newtheorem{remark}{Remark}
\newtheorem{example}{Example}
\newtheorem{problem}{Problem}
\newtheorem{algorithm}{Algorithm}

\begin{document}
%
\title{Technical Report: NUS-ACT-10-001-Ver.3: Synchronized Task
Decomposition for Cooperative Multi-agent Systems\\
Preprint, Accepted by Automatica (entitled as Guaranteed Global
Performance Through Local Coordinations)}
\author{Mohammad~Karimadini,
        and~Hai~Lin,
\thanks{M. Karimadini and H. Lin are both from the Department of Electrical and Computer Engineering, National University of Singapore,
Singapore. Corresponding author, H. Lin {\tt\small elelh@nus.edu.sg}
} }

\maketitle \thispagestyle{empty} \pagestyle{empty}

\begin{abstract}
It is an amazing fact that remarkably complex behaviors could emerge
from a large collection of very rudimentary dynamical agents through
very simple local interactions. However, it still remains elusive on
how to design these local interactions among agents so as to achieve
certain desired collective behaviors. This paper aims to tackle this
challenge and proposes a divide-and-conquer approach to guarantee
specified global behaviors through local coordination and control
design for multi-agent systems. The basic idea is to decompose the
requested global specification into subtasks for each individual
agent. It should be noted that the decomposition is not arbitrary.
The global specification should be decomposed in such a way that the
fulfillment of these subtasks by each individual agent will imply
the satisfaction of the global specification as a team. Formally, a
given global specification can be represented as an automaton $A$,
while a multi-agent system can be captured as a parallel distributed
system. The first question needs to be answered is whether it is
always possible to decompose a given task automaton $A$ into a
finite number of sub-automata $A_i$, where the parallel composition
of these sub-automata $A_i$ is bisimilar to the automaton $A$.
First, it is shown by a counterexample that not all specifications
can be decomposed in this sense. Then, a natural follow-up question
is what the necessary and sufficient conditions should be for the
proposed decomposability of a global specification. The main part of
the paper is set to answer this question. The case of two
cooperative agents is investigated first, and necessary and
sufficient conditions are presented and proven. Later on, the result
is generalized to the case of arbitrary finite number of agents, and
a hierarchical algorithm is proposed, which is shown to be a
sufficient condition. Finally, a cooperative control scenario for a
team of three robots is developed to illustrate the task
decomposition procedure.

\end{abstract}

%
%

\section{Introduction}
Strongly driven by its great potential in both civilian and
industrial applications \cite{Lesser1999, Lima2005}, multi-agent
systems has rapidly emerged as a hot research area at the
intersection of control, communication and computation
\cite{Choi2009, Kazeroon2009, Georgios2009, Ji2009}. The key issue
in multi-agent system design is how to explicitly set the local
interaction rules such that certain desirable global behaviors can
be achieved by the team of cooperative agents\cite{Crespi2008}.
Although it is known that sophisticated collective behaviors could
emerge from a large collection of very elementary agents through
simple local interactions, we still lack knowledge on how to change
these rules to achieve or avoid certain global behaviors. As a
result, it still remains elusive on how to design these local
interaction rules so as to make sure that they, as a group, can
achieve the specified requirements. In addition, we would like to
point out that the desired collective behavior for a group of mobile
agents could be very complicated, and might involve the coordination
of many distributed independent and dependent modules in a parallel
and environmental awareness manner. Hence, it poses new challenges
that go beyond the traditional path planning, output regulation, or
formation control\cite{Tabuada2006, Belta2007, Kloetzer2010,
Georgios2009}.

This paper aims to propose a decomposition approach applicable in
divide-and-conquer design for cooperative multi-agent systems so as
to guarantee the desired global behaviors. The core idea is to
decompose a global specification into sub-specifications for
individual agents, and then design local controllers for each agent
to satisfy these local specifications, respectively. The
decomposition should be done in such a way that the global behavior
is achieved provided that all these sub-specifications are held true
by individual agents. Hence, the global specification is guaranteed
by design. In order to perform this idea, several questions are
required to be answered, such as how to describe the global
specification and subtasks in a succinct and formal way; how to
decompose the global specification; whether it is always possible to
decompose, and if not, what are the necessary and sufficient
conditions for decomposability.

To formally describe the specification, a deterministic finite
automaton is chosen here to represent the global specifications for
multi-agent systems, due to its expressibility for a large class of
tasks\cite{Cassandras2008, Kumar1999}, its similarity to our human
logical commands, and its connection to the temporal logic
specifications\cite{Gotzhein1992, Wolper2002}. Accordingly, we will
focus on the logical behavior of a multi-agent system and model its
collective behavior through parallel composition \cite{Mukund2002}.
It is assumed that each agent is equipped with a local event set,
containing both private events and common events (shared with other
agents). Furthermore, it is assumed that the global task automaton
is defined over the union of all agents' events. Then, the
decomposition problem can be stated as follows. Given the global
desired behavior represented as a deterministic automaton, how to
obtain the local task automata from it such that the composition of
the obtained local task automata is equivalent to the original
global task automaton?

Since the global task automaton is defined over the union of all
agents' event sets, a reasonable way to obtain the local task
automata is through natural projections with respect to each agent's
event set. Namely, the agent will ignore the transitions marked by
the events that are not in its own event set, i.e., blinds to these
moves. The obtained automaton will be a sub-automaton of the global
task automaton by deleting all the moves triggered by blind events
of the agent. Given a task automaton and sets of local events, it is
always feasible to do the projection operation, but the question is
whether the obtained sub-automata preserve the required
specifications in the sense that the fulfillment of each agent with
its corresponding sub-task automaton will imply the satisfaction of
the global specification as a group. Unfortunately, by a simple
counterexample, it can be shown that the answer is not always. Then,
a natural follow-up question is what the necessary and sufficient
condition should be for the proposed decomposability of a global
specification. The main part of the paper is set to answer this
question.

Similar automaton decomposition problem has been studied in the
computer science literature. Roughly speaking, two different classes
of problems have been studied, so far. The first problem is to
design the event distribution so as to make the automaton
decomposable, which is typically studied in the context of
concurrent systems. For example, \cite{Morin1998} characterized the
conditions for decomposition of asynchronous automata in the sense
of isomorphism based on the maximal cliques of the dependency graph.
The isomorphism equivalence used in \cite{Morin1998} is however a
strong condition, in the sense that two isomorphic automata are
bisimilar but not vise versa \cite{Cassandras2008}. In many
applications bisimulation relation suffices to capture the
equivalence relationship. On the other hand, the second class of
problems assumes that the distribution of the global event set is
given and the objective is to find conditions on the automaton such
that it is decomposable. This is usually called synthesis modulo
problem \cite{Mukund2002} that can be investigated under three types
of equivalence: isomorphism, bisimulation and language equivalence.
Bisimulation synthesis modulo for a global automaton was addressed
in \cite{Castellani1999}, by introducing necessary and sufficient
conditions for automaton decomposition based on language product of
the automaton and determinism of its bisimulation quotient.
Obtaining the bisimulation quotient, however, is generally a
difficult task, and the condition on language product relies on
language separability\cite{Willner1991}, which is another form of
decomposability. These problems motivate us to develop new necessary
and sufficient conditions that can characterize the decomposability
based on the investigation of events and strings in the given
automaton.

In this paper, we identify conditions on the global specification
automaton in terms of its private and common events for the proposed
decomposability, which are shown to be necessary and sufficient for
the case of two agents. Later on, the result is generalized to the
case of arbitrary finite number of agents, and a hierarchical
algorithm is proposed, which is shown to be a sufficient condition.
Furthermore, it is shown that if the global task is decomposable,
then designing the local controller for each agent to satisfy its
corresponding sub-task will lead the entire multi-agent system to
achieve the global specification. To illustrate the decomposition
approach, a coordination and control scenario has been developed and
implemented on a team of three robots.

The rest of the paper is organized as follows. Preliminary lemmas,
notations, definitions and problem formulation are represented in
Section \ref{PROBLEM FORMULATION}. Section \ref{TASK DECOMPOSITION
FOR TWO AGENTS} introduces the necessary and sufficient conditions
for decomposition of an automaton with respect to parallel
composition and two local event sets. The algorithm for the
hierarchical extension of the automaton decomposition is given in
Section \ref{HIERARCHICAL DECOMPOSITION}. To illustrate the task
decomposition, an implementation result is given on a cooperative
multi-robot system example in Section \ref{EXAMPLE}. Finally, the
paper concludes with remarks and discussions in Section
\ref{CONCLUSION}. The proofs of lemmas are given in the Appendix.

\section{Problem formulation}\label{PROBLEM FORMULATION}

We first recall the definition of an automaton \cite{Kumar1999}.
\begin{definition}(Automaton)
An automaton is a tuple $A = \left(Q, q_0, E, \delta \right)$
consisting of
\begin{itemize}
\item a set of states $Q$; \item the initial state
$q_0\in Q$; \item a set of events $E$ that causes transitions
between the states, and \item a transition relation $\delta
\subseteq Q \times E\times Q$ such that $(q, e, q^{\prime})\in
\delta$ if and only if $\delta(q, e)=q^{\prime}$ (or
$q\overset{e}{\underset{}\rightarrow } q^{\prime}$ ).
\end{itemize}
In general, automaton has also an argument $Q_m$ which is the set of
final or marked states, the states are marked when it is desired to
attach a special meaning to them, such as accomplishing. This
argument is dropped when the notion of marked state is not of
interest or it is clear from the content.
\end{definition}
As $\delta$ is a partial map from $Q \times E$ into $Q$, in general,
not all states are reachable from the initial state. The accessible
portion of the automaton is defined as
\begin{definition}(Accessible Operator; \cite{Cassandras2008})\label{Accessible Operator}
Consider an automaton $A = \left(Q, q_0, E, \delta \right)$. The
operator $Ac(.)$ excludes the states and their attached transitions
that are not reachable from the initial state, and is defined as
$Ac(A) = \left(Q_{ac}, q_0, E, \delta_{ac} \right)$ with
$Q_{ac}=\{q\in Q|\exists s\in E^*, \delta (q_0, s)=q\}$ and
$\delta_{ac}=\delta|Q_{ac}\times E\rightarrow Q_{ac}$, restricting
$\delta$ to the smaller domain of $Q_{ac}$. As $AC(.)$ has  no
effect on the behavior of the automaton, from now on we assume $A =
Ac(A)$.
\end{definition}

 The transition relation can be extended to a finite string of
events, $s\in E^*$, where $E^*$ stands for $Kleene-Closure$ of $E$
(the set of all finite strings over elements of $E$), as follows
$\delta(q,\varepsilon)=q$, and $\delta(q,se)=\delta(\delta(q,s),e)$
for $s\in E^*$ and $e\in E$. We focus on deterministic task automata
that are simpler to be characterized, and cover a wide class of
specifications. The qualitative behavior of a deterministic discrete
event system (DES) is described by the set of all possible sequences
of events starting from initial state. Each such a sequence is
called a string, and the collection of strings represents the
language generated by the automaton, denoted by $L(A)$. Given a
language $L$, $\overline{L}\subseteq E^*$ is the prefix-closure of
$L$ defined as $\overline{L} = \{s\in E^*|\exists t\in E^*, st\in
L\}$, consisting of all prefixes of all the strings in $L$. The
existence of a transition over string $s\in E^*$ from a state $q\in
Q$, is denoted by $\delta(q, s)!$, and considering a language $L$,
by $\delta(q, L)!$ we mean $\forall \omega \in L: \delta(q,
\omega)!$.

To describe the decomposability condition in the main result and
during the proofs, we define successive event pair and adjacent
event pair as follows.
\begin{definition}(Successive event pair)
Two events $e_1$ and $e_2$ are called successive events if $\exists
q\in Q: \delta(q,e_1)! \wedge \delta(\delta(q,e_1),e_2)!$ or
$\delta(q,e_2)! \wedge \delta(\delta(q,e_2),e_1)!$.
\end{definition}
\begin{definition}(Adjacent event pair)
Two events $e_1$ and $e_2$ are called adjacent events if $\exists
q\in Q:\delta(q,e_1)! \wedge \delta(q,e_2)!$.
\end{definition}

To compare the task automaton and its decomposed automata, we use
the simulation and bisimulation relations \cite{Cassandras2008}.
\begin{definition}(Simulation and Bisimulation)
\label{simulation} Consider two automata $A_i=( Q_i, q_i^0$, $E,
\delta _i)$, $i=1, 2$. A relation $R\subseteq Q_1 \times Q_2$ is
said to be a simulation relation from $A_1$ to $A_2$ if
\begin{enumerate}
    \item
$(q_1^0, q_2^0) \in R$
    \item
$\forall\left( {q_1 ,q_2 } \right) \in R, \delta_1(q_1, e)= q'_1$,
then $\exists q_2^{\prime}\in Q_2$ such that $\delta_2(q_2, e)=q'_2,
\left( {q'_1 ,q'_2 } \right) \in R$.
\end{enumerate}

If $R$ is defined for all states and all events in $A_1$, then $A_1$
is said to be similar to $A_2$ (or $A_2$ simulates $A_1$), denoted
by $A_1\prec A_2$ \cite{Cassandras2008}.

If $A_1\prec A_2$, $A_2\prec A_1$ and $R$ is symmetric then $A_1$
and $A_2$ are said to be bisimilar (bisimulate each other), denoted
by $A_1\cong A_2$ \cite{Zhou2006}.

In general, bisimilarity implies languages equivalence but the
converse does not necessarily hold true \cite{Alur2000}.
\end{definition}

\begin{definition} (Isomorphism, \cite{Lawson07})\label{Isomorphism}
Isomorphism is one of the strongest equivalences relations between
automata. Two automata $A_i=( Q_i, q_i^0, E, \delta _i)$, $i=1, 2$,
are said to be isomorphic, if there exists an isomorphism $\theta$
from $A_1$ to $A_2$ defined as a bijective function $\theta: Q_1\to
Q_2$ such that $\theta(q_1^0) = q_2^0$, and $\theta(\delta_1(q, e))
= \delta_2(\theta(q), e)$, $\forall q\in Q_1, e\in E$.

By this definition, two isomorphic automata are bisimilar, but
bisimilar automata are not necessarily isomorphic (see Example
\ref{bisimilarityEX}).
\end{definition}

In this paper, we assume that the task automaton $A_S$ and the sets
of local events $E_i$ are all given. It is further assumed that
$A_S$ is deterministic  while its event set $E$ is obtained by the
union of local event sets, i.e., $E=\cup_i E_i$. The problem is to
check whether the task automaton $A_S$ can be decomposed into
sub-automata $A_{S_i}$ on the local event sets $E_i$, respectively,
such that the collection of these sub-automata $A_{S_i}$ is somehow
equivalent  to $A_S$. The equivalence is in the sense of
bisimilarity as defined above, while the composition process for
these sub-automata $A_{S_i}$ could be in the usual sense of parallel
composition as defined below. Parallel composition is used to model
the interactions between automata and represent the logical behavior
of multi-agent systems. Parallel composition is formally defined as
\begin{definition} (Parallel Composition \cite{Kumar1999}) \label{parallel
composition}Let $A_i=\left( Q_i,q_i^0,E_i,\delta _i\right)$,
$i=1,2$, be automata. The parallel composition (synchronous
composition) of $A_1$ and $A_2$ is the automaton $A_1||A_2=\left(Q =
Q_1 \times Q_2, q_0 = (q_1^0, q_2^0), E = E_1 \cup E_2,
\delta\right)$, with $\delta$ defined as\\
   $\forall (q_1, q_2)\in Q, e\in E:\delta(\left(q_1, q_2), e\right)=
    \left\{
\begin{array}{ll}
    \left(\delta_1(q_1, e), \delta_2(q_2, e)\right), & \hbox{if $\left\{\begin{array}{ll}
    \delta _1(q_1,e)!, \delta _2(q_2,e)!, \\
     e\in E_1 \cap E_2
\end{array}\right.$};\\
    \left(\delta_1(q_1, e), q_2\right), & \hbox{if $\delta _1(q_1,e)!, e\in E_1 \backslash E_2$;} \\
    \left(q_1, \delta_2(q_2, e)\right), & \hbox{if $\delta _2(q_2,e)!, e\in E_2 \backslash E_1$;} \\
    \hbox{undefined,} & \hbox{otherwise} \\
\end{array}\right.$.

The parallel composition of $A_i$, $i=1,2,...,n$ is called parallel
distributed system, and is defined based on the associativity
property of parallel composition \cite{Cassandras2008} as
$\overset{n}{\underset{i=1}{\parallel\ } }A_i=A_1\parallel\
...\parallel\   A_n=A_1\parallel \left(A_2
\parallel \left( \cdots \parallel \left( A_{n-1}\parallel
A_n \right)\right)\right)$.
\end{definition}

A reasonable guess for task automaton decomposition is to use
natural projections with respect to agents' event set. Natural
projection over strings is denoted by $p_{E_i}=p  _i:E^*\rightarrow
E_i^*$, takes a string from the event set $E$ and eliminates events
in it that do not belong to the event set $E_i\subseteq E$. The
natural projection is formally defined on the strings as
\begin{definition}(Natural Projection on String, \cite{Cassandras2008})
Consider a global event set $E$ and its local event sets $E_i$,
$i=1,2,...,n$, with $E=\overset{n}{\underset{i=1}{\cup} } E_i$.
Then, the natural projection $p_i:E^*\rightarrow
E_i^*$ is inductively defined as \\
 $p_i(\varepsilon)=\varepsilon$; \\
 $\forall s\in E^*, e\in E:p_i(se)=\left\{
  \begin{array}{ll}
  p_i(s)e & \hbox{if $e\in E_i$;} \\
  p_i(s) & \hbox{otherwise.}
  \end{array}
\right.$\\
\end{definition}

The natural projection is also defined on automata as
$P_i(A_S):A_S\rightarrow A_{S_i}$, where, $A_{S_i}$ are obtained
from $A_S$ by replacing its events that belong to $E\backslash E_i$
by $\tau$-moves (representing silent or unobservable transitions),
and then, merging the $\tau$-related states. The  $\tau$-related
states form equivalent classes defined as follows.
\begin{definition}(Equivalent class of states, \cite{Morin1998})\label{Equivalent class of states}
Consider an automaton $A_S=(Q, q_0, E, \delta)$ and local event sets
$E_i$, $i=1,2,...,n$, with $E=\overset{n}{\underset{i=1}{\cup} }
E_i$. Then, the relation $\sim E_i$ (or $\sim_i$) is the least
equivalence relation on the set $Q$ of states such that
$\delta(q,e)=q^{\prime}\wedge e\notin E_i\Rightarrow
q\sim_{E_i}q^{\prime}$, and $[q]_{E_i}=[q]_i$ denotes the
equivalence class of $q$ defined on $\sim E_i$. In this case, $q$
and $q^{\prime}$ are said to be $\tau$-related.
\end{definition}
The natural projection is then formally defined on an automaton as
follows.
\begin{definition}(Natural Projection on Automaton)\label{Natural Projection on Automaton}
Consider an automaton $A_S=(Q, q_0, E, \delta)$ and local event sets
$E_i$, $i=1,2,...,n$, with $E=\overset{n}{\underset{i=1}{\cup} }
E_i$. Then, $P_i(A_S)=(Q_i=Q_{/\sim E_i}, [q_0]_{E_i}, E_i,
\delta_i)$, with $\delta_i([q]_{E_i}, e)=[q^{\prime}]_{E_i}$ if
there are states $q_1$ and $q_1^{\prime}$ such that
$q_1\sim_{E_i}q$, $q_1^{\prime}\sim_{E_i}q^{\prime}$, and
$\delta(q_1, e)=q^{\prime}_1$.
\end{definition}

The following example elaborates the concept of natural projection
on a given automaton.

\begin{example}\label{Natural Projection Example}
Consider an automaton $A_S$: \xymatrix@C=0.5cm{
     \ar[r]&  \bullet \ar[r]^{a}\ar[d]^{e_2} &  \bullet
     \ar[d]^{e_1}&
     \\
     & \bullet \ar[r]^{e_4} &  \bullet \ar[r]^{b}&\bullet
       }
with the event set $E=E_1\cup E_2$ and local event sets $E_1=\{a, b,
e_1\}$, $E_2=\{a, b, e_2, e_4\}$. The natural projections of $A_S$
into $E_1$ is obtained as $P_1(A_S)$: \xymatrix@C=0.5cm{
     \bullet &  \check{\bullet} \ar@/_/[r]_{a}\ar[l]_{b} &  \bullet \ar[l]_{e_1}
         }
by replacing $\{e_2, e_4\}\in E \backslash E_1$ with $\tau$ and
merging the $\tau$-related states. Similarly, the projection
$P_2(A_S)$ is obtained as $P_2(A_S)$:
 \xymatrix@C=0.5cm{
     \ar[r]&  \bullet \ar[r]^{e_2} \ar@/_/[rr]_{a}&  \bullet \ar[r]^{e_4}&
     \bullet \ar[r]^{b}&
     \bullet}.
\end{example}

To investigate the interactions of transitions in two automata,
particularly in $P_1(A_S)$ and $P_2(A_S)$, the interleaving of
strings is defined, based on the path automaton as follows.
\begin{definition}(Path Automaton)\label{Path Automaton}
A sequence $q_1\overset{e_1}{\underset{}\rightarrow }
q_2\overset{e_2}{\underset{}\rightarrow
}...\overset{e_n}{\underset{}\rightarrow }q_n$ is called a path
automaton, characterized by its initial state $q_1$ and string $s =
e_1e_2...e_n$, denoted by $PA(q_1, s)$, and defined as $PA(q_1, s)=
(\{q_1,...,q_n\}, \{q_1\}$, $\{e_1,...,e_n\}$,$ \delta_{PA})$, with
$\delta_{PA}(q_i, e_i) = q_{i+1}$, $i = 1,...,n-1$.
$PA(q_1^{\prime}, s^{\prime})$ is defined, similarly.
\end{definition}

\begin{definition}(Interleaving)\label{Interleaving}
Given two sequences $q_1\overset{e_1}{\underset{}\rightarrow }
q_2\overset{e_2}{\underset{}\rightarrow
}...\overset{e_n}{\underset{}\rightarrow }q_n$ and
$q_1^{\prime}\overset{e_1^{\prime}}{\underset{}\rightarrow }
q_2^{\prime}\overset{e_2^{\prime}}{\underset{}\rightarrow
}...\overset{e_m^{\prime}}{\underset{}\rightarrow }q_m^{\prime}$,
the interleaving of their corresponding strings, $s = e_1e_2...e_n$
and $s^{\prime} = e_1^{\prime}e_2^{\prime}...e_m^{\prime}$, is
denoted by $s|s^{\prime}$, and defined as $s|s^{\prime}= L\{PA(q_1,
s)||PA(q_1^{\prime}, s^{\prime})\}$.
\end{definition}

\begin{example}
Consider three strings $s_1 = e_1a$, $s_2 = ae_2$ and $s_3 = ae_1$.
Then the interleaving of $s_1$ and $s_2$ is $s_1|s_2 =
\overline{e_1ae_2}$ while the interleaving of two strings $s_2$ and
$s_3$ becomes $s_2|s_3 = \overline{\{ae_1e_2, ae_2e_1\}}$.
\end{example}

Based on these definitions, we are now ready to formally define the
decomposability of an automaton with respect to parallel composition
and natural projections as follows.

\begin{definition}(Automaton decomposability, or bisimulation synthesis modulo \cite{Mukund2002})
 A task automaton $A_S$ with the event set $E$ and local event sets
 $E_i$, $i=1,..., n$, $E = \overset{n}{\underset{i=1}{\cup} } E_i$, is
said to be decomposable with respect to parallel composition and
natural projections $P_i: A_S\rightarrow P_i \left( A_S \right)$,
$i=1,\cdots, n$, if $\overset{n}{\underset{i=1}{\parallel} } P_i
\left( A_S \right) \cong A_S$.
\end{definition}

Now, let us see a motivating example to illustrate the decomposition
procedure.
\begin{example}\label{bisimilarityEX}
Consider the task automaton $A_S$ and its local event sets in
Example \ref{Natural Projection Example}. This automaton is
decomposable with respect to parallel composition and natural
projections, since $A_S\cong P_1(A_S)||P_2(A_S)$, leading to
$L(A_S)=L(P_1(A_S)||P_2(A_S))=\{\varepsilon, a, a e_1, ae_1b, e_2,
e_2e_4, e_2e_4b\}$.

Two automata\\
 \xymatrix@R=0.1cm{
                &    &     *+[o][F]{q_1}  \ar[dr]^{e_2} &   \\
\ar[r]&  *+[o][F]{q_0} \ar[ur]^{e_1} \ar[dr]_{e_2} & &  *+[o][F]{q_2}t  \ar[r]^{a}& *+[o][F]{q_3}  \\
              &   & *+[o][F]{q_4} \ar[ur]_{e_1}&
                }  and
 \xymatrix@R=0.1cm{
                &    & *+[o][F]{q_1^{\prime}} \ar[r]^{e_2}  & *+[o][F]{q_2^{\prime}}\ar[r]^{a}& *+[o][F]{q_3^{\prime}}  \\
\ar[r]&  *+[o][F]{q_0^{\prime}} \ar[ur]^{e_1} \ar[dr]_{e_2} & &  &               \\
             &    & *+[o][F]{q_4^{\prime}}  \ar[r]_{e_1} & *+[o][F]{q_5^{\prime}}\ar[r]_{a}& *+[o][F]{q_6^{\prime}}
                }  \\ with $E=E_1\cup E_2$, $E_1=\{a, e_1\}$, $E_2=\{a, e_2\}$ are other examples of decomposable automata.

Note that, these two automata are bisimilar (with the bisimulation
relation $R = \{(q_0, q_0^{\prime}),(q_1, q_1^{\prime}),\\(q_2,
q_2^{\prime}),(q_3, q_3^{\prime}),(q_4, q_4^{\prime}),(q_2,
q_5^{\prime}),(q_3, q_6^{\prime})\}$), but not isomorphic. The first
one is concrete (decomposable in the sense of isomorphism;
satisfying $FD$ (forward diamond) and $ID$ (independent diamond)) as
well as decomposable (in the sense of bisimulation). The latter
automaton, on the other hand, is decomposable, but not concrete.
\end{example}

Examples \ref{bisimilarityEX} shows a decomposable task automaton;
however, to deal with the top - down design we need to understand
whether any task automaton is decomposable or not.
\begin{problem}\label{solvability}
Given a deterministic task automaton $A_S$ and local event sets
$E_i$, $i=1,\cdots, n$, is it always possible to decompose $A_S$
with respect to parallel composition and natural projections $P_i$,
$i=1,\cdots, n$?
\end{problem}

To answer this question, we examine the following example.

\begin{example}\label{SimilarityEX}
Consider an automaton $A_S$: \\ \xymatrix@C=0.5cm{
     \ar[r]&  \bullet \ar[r]^{e_1} &  \bullet \ar[r]^{e_2} &  \bullet
            }, with local event sets $E_1=\{e_1\}$ and $E_2=\{e_2\}$.
             The parallel composition of $P_1(A_S):  \xymatrix@C=0.5cm{
     \ar[r]&  \bullet \ar[r]^{e_1}&  \bullet
            }$
            and $P_2(A_S):  \xymatrix@C=0.5cm{
     \ar[r]&  \bullet \ar[r]^{e_2}&  \bullet
            }$ is
     $P_1(A_S)||P_2(A_S)$:  \xymatrix@R=0.1cm{
                \ar[r]&\bullet \ar[r]^{e_1}\ar[dr]_{e_2}&\bullet   \ar[r]^{e_2}&\bullet \\
             && \bullet \ar[ur]_{e_1} }.

One can observe that  $A_S\prec P_1(A_S)||P_2(A_S)$ but $
P_1(A_S)||P_2(A_S)\nprec A_S$ leading to $L(A_S)=\{\varepsilon, e_1,
e_1e_2\}\subseteq L(P_1(A_S)||P_2(A_S))=\{\varepsilon, e_1, e_1e_2,
e_2, e_2e_1\}$, but $L(P_1(A_S)||P_2(A_S))\nsubseteq L(A_S)$.
Therefore, $A_S$ is not decomposable with respect to parallel
composition and natural projections $P_i$, $i=1,2$.
\end{example}

Therefore, not all automata are decomposable with respect to
parallel composition and natural projections. Then, a natural
follow-up question is what makes an automaton decomposable. It can
be formally stated as follows.

\begin{problem}\label{solution} Given a deterministic task automaton $A_S$
and local event sets $E_i$, $i=1,\cdots, n$, what are the necessary
and sufficient conditions that $A_S$ is decomposable with respect to
parallel composition and natural projections $P_i: A_S\rightarrow
P_i \left( A_S \right)$, $i=1,\cdots, n$, such that
$\overset{n}{\underset{i=1}{\parallel} } P_i \left( A_S \right)
\cong A_S$?
\end{problem}

This problem will be addressed in the following two sections. In the
next section, we will deal with the case of two agents and present
necessary and sufficient conditions for decomposability. Then, the
result is generalized in Section \ref{HIERARCHICAL DECOMPOSITION} to
a finite number of agents, as a sufficient condition.

\section{Task decomposition for two agents}\label{TASK DECOMPOSITION FOR TWO AGENTS}

%
%
%
%

In order for $A_S\cong P_1(A_S)||P_2(A_S)$, from the definition of
bisimulation, it is required to have $A_S\prec P_1(A_S)||P_2(A_S)$;
$P_1(A_S)||P_2(A_S)\prec A_S$, and the simulation relations are
symmetric. These requirements are provided by the following three
lemmas. Firstly, the following simulation relationship always holds
true.
\begin{lemma}\label{similarity relation of parallel}
Consider any deterministic automaton $A_S$ with event set $E =
E_1\cup E_2$, local event sets $E_i$, and natural projections $P_i$,
$i=1,2$. Then $A_S \prec P_1(A_S)||P_2(A_S)$.
\end{lemma}
\begin{proof}
See the Appendix for proof.
\end{proof}

This lemma shows that, in general, $P_1(A_S)||P_2(A_S)$ simulates
$A_S$. The similarity of $P_1(A_S)||P_2(A_S)$ to $A_S$, however, is
not always true (see Example \ref{SimilarityEX}), and needs some
conditions as stated in the following lemma.

\begin{lemma} \label{reverse similarity relation of parallel}
 Consider a deterministic  automaton $A_S =(Q, q_0, E=E_1\cup E_2, \delta)$
and natural projections $P_i:A_S\to P_i(A_S)$, $i=1,2$. Then,
$P_1(A_S)||P_2(A_S)\prec A_S$ if and only if $A_S$ satisfies the
following conditions: $\forall e_1 \in E_1\backslash E_2, e_2 \in
E_2\backslash E_1, q\in Q$, $s\in E^*$:
\begin{itemize}\item $DC1: [\delta(q,e_1)!\wedge
\delta(q,e_2)!]\Rightarrow [\delta(q, e_1e_2)! \wedge \delta(q,
e_2e_1)!]$;
\item $DC2: \delta(q, e_1e_2s)!\Leftrightarrow \delta(q, e_2e_1s)!$, and
 \item $DC3:
\forall s, s^{\prime} \in E^*$, sharing the same first appearing
common event $a\in E_1 \cap E_2$, $s\neq s^{\prime}$, $q\in Q$:
$\delta(q, s)! \wedge \delta(q, s^{\prime})! \Rightarrow \delta(q,
p_1(s)|p_2(s^{\prime}))! \wedge \delta(q, p_1(s^{\prime})|p_2(s))!$.
\end{itemize}
\end{lemma}
\begin{proof}
See the Appendix for proof.
\end{proof}

Next, we need to show that the two simulation relations $R_1$ (for
$A_S \prec P_1(A_S)||P_2(A_S)$) and $R_2$ (for
$P_1(A_S)||P_2(A_S)\prec A_S$), defined by the above two lemmas, are
symmetric.

\begin{lemma}\label{symmetric}
Consider an automaton $A_S=(Q, q_0, E=E_1\cup E_2, \delta)$ with
natural projections $P_i:A_S\to P_i(A_S)$, $i=1,2$. If $A_S$ is
deterministic, $A_S\prec P_1(A_S)||P_2(A_S)$ with the simulation
relation $R_1$ and $P_1(A_S)||P_2(A_S)\prec A_S$ with the simulation
relation $R_2$, then $R^{-1}_1 = R_2$ (i.e., $\forall q\in Q$, $z\in
Z$: $(z, q)\in R_2 \Leftrightarrow (q, z)\in R_1$) if and only if
$DC4$: $\forall i\in\{1, 2\}$, $x, x_1, x_2 \in Q_i$, $x_1\neq x_2$,
$e\in E_i$, $t\in E_i^*$, $\delta_i (x, e)=  x_1$, $\delta_i (x, e)=
x_2$: $\delta_i (x_1, t)! \Leftrightarrow \delta_i(x_2, t)!$.
\end{lemma}
\begin{proof}
See the proof in the Appendix.
\end{proof}

Based on these lemmas, the main result on task automaton
decomposition is given as follows.
\begin{theorem}\label{Task Automaton Decomposition}
A deterministic  automaton $A_S = (Q, q_0, E=E_1\cup E_2, \delta)$
is decomposable with respect to parallel composition and natural
projections $P_i: A_S\to P_i(A_S)$, $i=1,2$, such that $A_S\cong
P_1(A_S)||P_2(A_S)$ if and only if $A_S$ satisfies the following
decomposability conditions (DC): $\forall e_1 \in E_1\backslash E_2,
e_2 \in E_2\backslash E_1, q\in Q$, $s\in E^*$,
\begin{itemize}\item $DC1$: $[\delta(q,e_1)!\wedge
\delta(q,e_2)!]\Rightarrow [\delta(q, e_1e_2)! \wedge \delta(q,
e_2e_1)!]$;
\item $DC2$: $\delta(q, e_1e_2s)!\Leftrightarrow \delta(q, e_2e_1s)!$, and \item $DC3$:
$\forall s, s^{\prime} \in E^*$, sharing the same first appearing
common event $a\in E_1 \cap E_2$, $s\neq s^{\prime}$, $q\in Q$:
$\delta(q, s)! \wedge \delta(q, s^{\prime})! \Rightarrow \delta(q,
p_1(s)|p_2(s^{\prime}))! \wedge \delta(q, p_1(s^{\prime})|p_2(s))!$;
\item $DC4$: $\forall i\in\{1, 2\}$, $x, x_1, x_2 \in Q_i$, $x_1\neq x_2$,
$e\in E_i$, $t\in E_i^*$, $\delta_i (x, e)=  x_1$, $\delta_i (x, e)=
x_2$: $\delta_i (x_1, t)! \Leftrightarrow \delta_i(x_2, t)!$.
\end{itemize}
\end{theorem}

\begin{proof}
According to Definition \ref{simulation}, $A_S\cong
P_1(A_S)||P_2(A_S)$ if and only if $A_S\prec P_1(A_S)||P_2(A_S)$
(that is always true due to Lemma \ref{similarity relation of
parallel}), $P_1(A_S)||P_2(A_S)\prec A_S$ (that it is true if and
only if $DC1$, $DC2$ and $DC3$ are true, according to Lemma
\ref{reverse similarity relation of parallel}) and the simulation
relations are symmetric, i.e., $R^{-1}_1 = R_2$(that for a
deterministic automaton $A_S$, when $A_S\prec P_1(A_S)||P_2(A_S)$
with simulation relation $R_1$ and $P_1(A_S)||P_2(A_S)\prec A_S$
with simulation relation $R_2$, due to Lemma \ref{symmetric},
$R^{-1}_1 = R_2$ holds true if and only if $DC4$ is satisfied).
Therefore, $A_S\cong P_1(A_S)||P_2(A_S)$ if and only if $DC1$,
$DC2$, $DC3$ and $DC4$ are satisfied.
\end{proof}

\begin{remark}\label{meaning of DC}
Intuitively, the decomposability condition $DC1$ means that for any
successive or adjacent pair of private events $(e_1, e_2)\in
\{(E_1\backslash E_2, E_2\backslash E_1)$, $(E_2\backslash E_1,
E_1\backslash E_2)\}$ (from different private event sets), both
orders $e_1e_2$ and $e_2e_1$ should be legal from the same state,
unless they are mediated by a common string.

Furthermore, $e_1e_2$ and $e_2e_1$ are not required to meet at the
same state (unlike $FD$ and $ID$ in \cite{Morin1998}); but due to
$DC2$, any string $s\in E^*$ after them should be the same, or in
other words, if $e_1$ and $e_2$ are necessary conditions for
occurrence of a string $s$, then any order of these two events would
be legal for such occurrence (see Example \ref{bisimilarityEX}).
Note that, as a special case, $s$ could be $\varepsilon$.

The condition $DC3$ means that if two strings $s$ and $s^{\prime}$
share the same first appearing common event, then any interleaving
of these two strings should be legal in $A_S$. This requirement is
due to synchronization of projections of these strings in $P_1(A_S)$
and $P_2(A_S)$.

The last condition, $DC4$, ensures the symmetry of mutual simulation
relations between $A_S$ and $P_1(A_S)||P_2(A_S)$. Given the
determinism of $A_S$, this symmetry is guaranteed when each local
task automaton bisimulates a deterministic automaton, leading to the
existence of a deterministic automaton that is bisimilar to
$P_1(A_S)||P_2(A_S)$. If the simulation relations are not symmetric,
then some of the sequences that are allowed in $A_S$ will be
disabled in $P_1(A_S)||P_2(A_S)$.

The notion of language decomposability \cite{Rudie1992} is
comparable with $DC2$ and means that any order of any successive
events in any string of the language specification should be legal,
or at least one of its projections (from the viewpoint of the
corresponding local observer) should be capable of distinguishing
this order. It also embodies a notion similar to $DC3$, stating that
the global languages specification should contain all possible
interleaving languages of all local languages. This notion, however,
is not capable of capturing the other two conditions on the decision
on the switch between adjacent transitions ($DC1$), and existence of
deterministic bisimilar automata to $P_1(A_S)$ and $P_2(A_S)$
($DC4$). The automaton decomposability conditions in this result, in
terms of bisimulation, besides checking the capability of local
plants on decision making on the orders of event ($DC2$), they
should also be capable of decision making on the switches ($DC1$),
and moreover, the synchronization of local task automata should not
lead to an illegal behavior($DC3$), and also ensures that like
$A_S$, $P_1(A_S)||P_2(A_S)$ also has a deterministic behavior
($DC4$).

The decomposability conditions can be then paraphrased as follows:
Any decision on order or switch between two transitions that cannot
be made locally (by at least one local controller) should not be
critical globally (any result of the decision should be allowed);
and interpretation of the global task by the team of local plants
should neither allow an illegal behavior (a string that is not in
global task automaton), nor disallow a legal behavior (a string that
appears in the global task automaton).
\end{remark}

The following four examples illustrate the decomposability conditions
for decomposable and undecomposable automata.

\begin{example} \label{Undecomposable DC1}
This example illustrates the concept of decision making on switching
between the events, mentioned in Remark \ref{meaning of DC}.
Furthermore, it shows an automaton that satisfies $DC2$, $DC3$ and
$DC4$, but not $DC1$, leading to undecomposability. The automaton
$A_S$: \xymatrix@R=0.1cm{
                \ar[r]&  \bullet \ar[r]^{e_1} \ar[dr]_{e_2}&\bullet  \\
             && \bullet                   }
             with local event sets $E_1=\{e_1\}$ and $E_2=\{e_2\}$,
             is not decomposable as the parallel composition of $P_1(A_S):  \xymatrix@C=0.5cm{
     \ar[r]&  \bullet \ar[r]^{e_1}&  \bullet
            }$
            and $P_2(A_S):  \xymatrix@C=0.5cm{
     \ar[r]&  \bullet \ar[r]^{e_2}&  \bullet
            }$ is
     $P_1(A_S)||P_2(A_S): \xymatrix@R=0.1cm{
                \ar[r]&\bullet \ar[r]^{e_1}\ar[dr]_{e_2}&\bullet   \ar[r]^{e_2}&\bullet .\\
             && \bullet \ar[ur]_{e_1} &&}$ which does not bisimulate
             $A_S$.
Here, $A_S$ is not decomposable with respect to parallel composition
and natural projections $P_i$, $i=1,2$, since two events $e_1\in
E_1\backslash E_2$ and $e_2\in E_2\backslash E_1$ do not respect
$DC1$, as none of the local plant takes in charge of decision making
on the switching between these two events. One can observe that, if
in this example $e_1\in E_1\backslash E_2$ and $e_2\in E_2\backslash
E_1$ were separated by a common event $a\in E_1 \cap E_2$, then
\xymatrix@R=0.1cm{
                \ar[r]&\bullet \ar[r]^{a}\ar[dr]_{e_1}&\bullet   \ar[r]^{e_2}&\bullet \\
             && \bullet  & &                }
             with local event sets $E_1=\{e_1, a\}$ and $E_2=\{e_2, a\}$, was decomposable,
             since the decision on the switch between $e_1$ and $a$ could be made in $E_1$
              and then $E_2$ could be responsible for the decision on the order of $a$ and $e_2$.
\end{example}

\begin{example}\label{Undecomposable DC2}
The automaton $A_S$ in Example \ref{bisimilarityEX} shows an
automaton that respects $DC1$, $DC3$ and $DC4$, but is
undecomposable due to violation of $DC2$. Here, $A_S$ is not
decomposable since none of the local plants take in charge of
decision making on the order of two events $e_1\in E_1\backslash
E_2$ and $e_2\in E_2\backslash E_1$. If $e_1\in E_1\backslash E_2$
and $e_2\in E_2\backslash E_1$ were separated by a common event
$a\in E_1 \cap E_2$, then the automaton \xymatrix@C=0.5cm{
     \ar[r]&  \bullet \ar[r]^{e_1} &  \bullet \ar[r]^{a} &  \bullet \ar[r]^{e_2} &  \bullet
            }
             with local event sets $E_1=\{e_1, a\}$ and $E_2=\{e_2, a\}$, was decomposable,
             since the decision on the orders of $e_1$ and $a$ and then $a$ and $e_2$ could be made in $E_1$
              and then $E_2$, subsequently.
As another example, consider an automaton $A_S$: \xymatrix@R=0.1cm{
                &    & \bullet  \ar[r]^{e_2}  & \bullet\ar[r]^{a}& \bullet  \\
\ar[r]&  \bullet \ar[ur]^{e_1} \ar[dr]_{e_2}                  \\
             &    & \bullet  \ar[r]_{e_1} &\bullet
                }  with $E_1=\{a, e_1\}$, $E_2=\{a, e_2\}$,
                leading to
$P_1(A_S)||P_2(A_S)$: \xymatrix@R=0.5cm{
                \bullet &\bullet  \ar[l]_{e_1} \ar[r]^{e_1}&\bullet &   \\
                \bullet \ar[u]^{e_2} \ar[d]_{e_2} &\grave{\bullet} \ar[l]_{e_1} \ar[r]^{e_1}\ar[u]^{e_2}  \ar[d]_{e_2} &\bullet  \ar[u]_{e_2}  \ar[d]^{e_2}& \\
                 \bullet &\bullet  \ar[l]_{e_1} \ar[r]^{e_1}&\bullet \ar[r]^{a}&\bullet
                 }. The transition $\delta_{||}(z_0, e_2e_1a)!$ in $P_1(A_S)||P_2(A_S)$,
but $\neg \delta(q_0, e_2e_1a)!$ in $A_S$. Therefore, $A_S$ is not
decomposable. If the lower branch was continued with a transition on
$a$ after $e_2e_1$, then the automaton was decomposable (See Example
\ref{bisimilarityEX}).

\end{example}
\begin{example}\label{Undecomposable DC3}
This example illustrates an automaton that satisfies $DC1$, $DC2$
and $DC4$, but it is undecomposable as it does not fulfil $DC3$,
since new strings appear in $P_1(A_S)||P_2(A_S)$ from the
interleaving of two strings in $P_1(A_S)$ and $P_2(A_S)$, but they
are not legal in $A_S$. Consider the task automaton
$A_S$:\xymatrix@R=0.1cm{
                &    &     \bullet  \ar[dr]^{e_2}    \\
\ar[r]&  \bullet \ar[ur]^{e_1} \ar[dr]_{e_2} \ar[dd]_{a}& &  \bullet  \ar[r]^{a}& \bullet  \\
              &   & \bullet \ar[ur]_{e_1}\\
              &\bullet \ar[r]_{e_2}&\bullet
                }
with $E_1=\{a, e_1\}$, $E_2=\{a, e_2\}$, leading to $P_1(A_S)\cong
\xymatrix@R=0.1cm{
                &    & \bullet  \ar[r]^{a}  & \bullet  \\
\ar[r]&  \bullet \ar[ur]^{e_1} \ar[dr]_{a} & &                 \\
             &    & \bullet &
                }$,
                $P_2(A_S)\cong \xymatrix@R=0.1cm{
              && \bullet  \ar[r]^{a}  & \bullet  \\
\ar[r]&  \bullet \ar[ur]^{e_2} \ar[dr]_{a} & &                 \\
             & & \bullet  \ar[r]_{e_2} & \bullet
                }$ and \\
 $P_1(A_S)||P_2(A_S)$: \xymatrix@C=0.5cm{
    \bullet & \ar[l]_{a} \bullet \ar[d]^{e_1} & \ar[l]_{e_2} \check{\bullet}
    \ar[r]^{a} \ar[d]^{e_1}&\bullet \ar[r]^{e_2}& \bullet \\
     \bullet & \ar[l]_{a}\bullet& \bullet  \ar[l]_{e_2} \ar[r]^{a}&\bullet\ar[r]^{e_2}& \bullet
       } that is not bisimilar to $A_S$ since two strings $e_2a$ and $e_1ae_2$ are newly
                generated, while they do not appear in $A_S$,
                although both $P_1(A_S)$ and  $P_2(A_S)$ are deterministic.

\end{example}

\begin{example}\label{Undecomposable DC3-truncation}
This example illustrates an automaton that satisfies $DC1$ and
$DC2$, and $DC3$, but is undecomposable as it does not fulfil $DC4$.
Consider the task automaton \\$A_S$: \xymatrix@R=0.1cm{
                &    & *+[o][F]{q_1}  \ar[r]^{a}  & *+[o][F]{q_2}\ar[r]^{b}& *+[o][F]{q_3}  \\
\ar[r]&  *+[o][F]{q_0} \ar[ur]^{e_1} \ar[dr]_{a}                \\
             &    &*+[o][F]{q_4}
                } \ with $E_1=\{a, b, e_1\}$, $E_2=\{a, b\}$,
               leading to\\
                $P_1(A_S)$:\xymatrix@R=0.1cm{
                &    & *+[o][F]{x_1}  \ar[r]^{a}  & *+[o][F]{x_2}\ar[r]^{b}& *+[o][F]{x_3}  \\
\ar[r]&  *+[o][F]{x_0} \ar[ur]^{e_1} \ar[dr]_{a}                \\
             &    &*+[o][F]{x_4}
                },
                 $P_2(A_S)$:\xymatrix@R=0.1cm{
                &    & *+[o][F]{y_1}  \ar[r]^{b}  &  *+[o][F]{y_2}  \\
\ar[r]&  *+[o][F]{y_0} \ar[ur]^{a} \ar[dr]_{a}                 \\
             &    & *+[o][F]{y_3}
                }, and \\
$P_1(A_S)||P_2(A_S)$: \xymatrix@R=0.1cm{
                &    & *+[o][F]{z_1}  \ar[r]^{a} \ar[dr]_{a} & *+[o][F]{z_2}\ar[r]^{b}& *+[o][F]{z_3}  \\
\ar[r]&  *+[o][F]{z_0} \ar[ur]^{e_1} \ar[dr]_{a} & &*+[o][F]{z_5}                 \\
             &    & *+[o][F]{z_4}
                }  which is not bisimilar to $A_S$. This task automaton $A_S$ satisfies $DC1$ and $DC2$ as contains no
                successive/adjacent transitions defined on different local event sets.
                But, it does not fulfil $DC4$, although any string in $T= \{
p_1(s)|p_2(s^{\prime}),\\ p_1(s^{\prime})|p_2(s)\}$ ($s$ and
$s^{\prime}$ are the top and bottom strings in $A_S$ and share the
first appearing common event $a\in E_1\cap E_2$), appears in $A_S$.
The reason is that there exists a transition on string $e_1a$ from
$z_0$ to $z_5$ that stops in $P_1(A_S)||P_2(A_S)$, whereas, although
$e_1a$ transits from $q_0$ in $A_S$, it does not stop afterwards.
This illustrate dissymmetry in simulation relations between $A_S$
and $P_1(A_S)||P_2(A_S)$. Note that $A_S \prec P_1(A_S)||P_2(A_S)$
with the simulation relation $R_1$ over all events in $E$, from all
states in $Q$ into some states in $Z$, as $R_1=\{(q_0, z_0), (q_1,
z_1),(q_2, z_2), (q_3, z_3),(q_4, z_4)\}$. Moreover, $
P_1(A_S)||P_2(A_S) \prec A_S$ with the simulation relation $R_2$
over all events in $E$, from all states in $Z$ into some states in
$Q$, as $R_2=\{(z_0, q_0), (z_1, q_1),(z_2, q_2),(z_3, q_3),(z_4,
q_4),(z_5, q_2)\}$. Therefore, although $A_S \prec
P_1(A_S)||P_2(A_S)$ and $ P_1(A_S)||P_2(A_S) \prec A_S$,
$P_1(A_S)||P_2(A_S) \ncong A_S$, since $\exists (z_5, q_2)\in R_2$,
whereas $(q_2, z_5)\notin R_1$. If for stoping of string $e_1a$ in
$P_1(A_S)||P_2(A_S)$, there was a state in $Q$ reachable from $q_0$
by $e_1a$ and stopping there, then we would have $\forall q\in Q,
z\in Z: (q,z)\in R_1\Leftrightarrow (z,q)\in R_2$ and
$P_1(A_S)||P_2(A_S) \cong A_S$.

It should be noted that the condition $DC4$ not only applies for
nondeterminism on common events, but also it requires any
nondeterminism on private event also to have a bisimilar
deterministic counterpart. For example, consider the task automaton
$A_S$: \xymatrix@R=0.1cm{ &&&\bullet\\
                \ar[r]&\bullet \ar[r]^{e_1}\ar[dr]_{e_2}&\bullet   \ar[r]_{e_2}\ar[ur]^{a}&\bullet\\
             && \bullet \ar[r]_{e_1}&\bullet} with $E_1 = \{e_1, a\}$, $E_2
             = \{e_2, a\}$. The parallel composition of $P_1(A_S)$: \xymatrix@R=0.1cm{
                \ar[r]&\bullet \ar[r]^{e_1}\ar[dr]_{e_1}&\bullet   \ar[r]^{a}&\bullet \\
             && \bullet } (with nondeterministci transition on private event $e_1$)
             and $P_2(A_S)$: \xymatrix@R=0.1cm{
                \ar[r]&\bullet \ar[r]^{a}\ar[dr]_{e_2}&\bullet   \\
             && \bullet } is $P_1(A_S)||P_2(A_S)$: \xymatrix@C=0.5cm{
             &&\bullet \ar[r]^{e_2}&\bullet\\
     \ar[r]&  \bullet \ar[r]^{e_2}\ar[ur]_{e_1}\ar[dr]^{e_1} &  \bullet
    \ar[ur]_{e_1}\ar[dr]^{e_1}\\
    &\bullet &  \bullet \ar[r]^{e_2} \ar[l]_{a}&  \bullet
       } which is not bisimilar to $A_S$.

 The automaton $A_S$: \xymatrix@R=0.1cm{
                &    & \bullet  \ar[r]^{a}  & \bullet  \\
\ar[r]&  \bullet \ar[ur]^{e_1} \ar[dr]_{a} & &                 \\
             &    & \bullet  \ar[r]_{e_2} & \bullet
                } with $E=E_1\cup E_2$, $E_1=\{a, e_1\}$, $E_2=\{a, e_2\}$,
                 $P_1(A_S)$:\xymatrix@R=0.1cm{
                &    & \bullet  \ar[r]^{a}  & \bullet  \\
\ar[r]&  \bullet \ar[ur]^{e_1} \ar[dr]_{a} & &                 \\
             &    & \bullet &
                } and $P_2(A_S)$: \xymatrix@R=0.1cm{
                &    & \bullet  \ar[r]^{e_2}  & \bullet  \\
\ar[r]&  \bullet \ar[ur]^{a} \ar[dr]_{a} & &                 \\
             &    & \bullet &
                } is an example of an undecomposable automaton that violates both
                $DC3$ and $DC4$. It violates $DC3$ since $\delta_{||}(z_0, e_1ae_2)!$ in
                $P_1(A_S)||P_2(A_S)$, but $\neg\delta(q_0,
                e_1ae_2)!$ in $A_S$, and it does not satisfy $DC4$ since $P_2(A_S)$
                is nondeterministic and is not bisimilar to a
                deterministic automaton, leading to a string in
                $P_1(A_S)||P_2(A_S)$ that $e_2$ is disallowed after
                $a$ while there in no such restriction in $A_S$.
If $A_S$ was $A_S$: \xymatrix@R=0.1cm{
                &    & \bullet  \ar[r]^{a}  & \bullet \ar[r]^{e_2}  & \bullet \\
\ar[r]&  \bullet \ar[ur]^{e_1} \ar[dr]_{a} & &                 \\
             &    & \bullet  \ar[r]_{e_2} & \bullet
                }, then $P_2(A_S)\cong \xymatrix@R=0.1cm{\ar[r]&  \bullet \ar[r]_{a} & \bullet \ar[r]_{e_2} &
                \bullet}$, and $A_S$ was decomposable.

\end{example}

\begin{remark}\label{language equivalent but not bisimilar}
Example \ref{Undecomposable DC3-truncation} also shows that the
determinism of $A_S$ does not reduce the bisimulation synthesis
problem to language equivalence synthesis problem. Note that here,
$A_S$ and $P_1(A_S)||P_2(A_S)$ are language equivalent, but not
bisimilar. The reason is that although $A_S$ is deterministic, and
$A_S \prec P_1(A_S)||P_2(A_S)$, $ P_1(A_S)||P_2(A_S) \prec A_S$, the
simulation relations are not symmetric due to existence of
nondeterministic strings in $P_1(A_S)||P_2(A_S)$ that can not be
replaced by a deterministic one. The nondeterminism in
$P_1(A_S)||P_2(A_S)$ is inherited from a nondeterminism in
$P_2(A_S)$. If $A_S$ was in the form of $A_S$: \xymatrix@R=0.1cm{
                &    & \bullet  \ar[r]^{a}  & \bullet\ar[r]^{b}& \bullet  \\
\ar[r]&  \bullet \ar[ur]^{e_1} \ar[dr]_{a}                \\
             &    &\bullet \ar[r]^{b}&\bullet
                }, then $P_1(A_S)\cong A_S$, $P_2(A_S)\cong \xymatrix@R=0.1cm{
    \ar[r]& \bullet  \ar[r]^{a}  & \bullet\ar[r]^{b}& \bullet}$ and
$P_1(A_S)||P_2(A_S)\cong A_S$.
\end{remark}

\section{Hierarchical decomposition}\label{HIERARCHICAL DECOMPOSITION}
The previous section showed the decomposition of an automaton with
respect to the parallel composition and two local event sets.
However, in practice, multi-agent systems are typically comprised of
many individual agents that work as a team. The proposed procedure
of decomposition can be generalized for more than two agents.
However, the problem becomes rapidly complex as the number of agents
increases. 
It is then advantageous to have a hierarchical decomposition method
to have only two individual event sets at a time for decomposition.
Consider a task automaton $A_S$ to be decomposed with respect to
parallel composition and individual event sets $E_i$, $i=1,2,...,n$,
so that $E=\overset{n}{\underset{i=1}{\cup}}E_i$. We propose the
following algorithm as a sufficient condition for hierarchical
decomposition of the given task automaton.

\begin{algorithm}\label{Hierarchical Decomposition Algorithm 1}(Hierarchical Decomposition
Algorithm)
\begin{enumerate}
\item $E=\overset{n}{\underset{i=1}{\cup}}E_i $, $\Sigma = \{E_1,...,
E_n\}$, $K =\{1,...,n\}$.
\item $i=1$, find $k\in K$ such that
$\Sigma _i = E_k\in \Sigma$, $\bar{\Sigma} _i =
\overset{}{\underset{j\in K\backslash k}{\cup}}E_j$, so that $A_S$
satisfies decomposability conditions $DC1$-$DC4$ in Theorem
\ref{Task Automaton Decomposition}, i.e., $A_S \cong P_{\Sigma
_i}(A_S)||P_{\bar{\Sigma} _i}(A_S)$.
\item $K = K\backslash k$, $\Sigma = \{E_j\}_{j\in K}$, $A_S = P_{\bar{\Sigma}_i}(A_S)$, $i = i+1$, go to Step $2$.
\item Continue until $i = n-1$ or no more decomposition is possible in
$i = m-1$, $m\leq n$. Then $\Sigma _m = \bar{\Sigma}_{m-1}$, and
hence, $A_S$ is decomposable with respect to parallel composition
and natural projections into $\{\Sigma_1, \cdots,
\Sigma_m\}\subseteq \Sigma$, if the algorithm proceeds up to $i =
m-1$.
\end{enumerate}
\end{algorithm}
\begin{remark}
The algorithm will terminate due to finite number of states and
local event sets. If the algorithm successfully proceeds to step
$n-1$, the automaton $A_S$ is decomposable and we obtain a complete
decomposition of the global specifications into subtasks for each
individual agent. However, it is unclear whether the algorithm can
 successfully terminate for any decomposable task automaton (necessity).
 The computational complexity of the algorithm in the worst case is of order
 $O(n^2(|E|^2\times |Q|\times \kappa+\overset{}{\underset{a\in E_1\cap E_2}{\Sigma}}|p_a(L(A_S))|^2))$, where
$\kappa = \overset{}{\underset{t\in L(A_S)}{max}}|t|$, assuming the
number of appearing events as the length of loops.
 In practice, during the iterations, $|E|$ is replaced by
 $\overset{|K|}{\underset{j=1}{\cup}}E_j$ which is decreasing with respect to iteration. Moreover, the second term in
 the complexity expression shows that the less number of common events
 and the less appearance of common events in $A_S$, the less
 complexity.
\end{remark}

Once the subtasks are obtained, the next step is the design of
controllers for each agent to achieve these subtasks respectively.
The following result shows that the fulfillment of the decomposed
subtasks will imply the global specifications for the multi-agent
systems. Before stating the theorem, following two lemmas are
presented to be used for the proof.
\begin{lemma}\label{associativity-parallel}
\begin{eqnarray} \label{associativity-hierarchicality}
P_{\Sigma_ 1}(A_S)\parallel P_{\Sigma_ 2}(A_S) \parallel \cdots
\parallel
P_{\Sigma_{m-1}}(A_S)\parallel P_{\Sigma_m}(A_S)\cong \nonumber \\
P_{\Sigma_1}(A_S)\parallel \left(P_{\Sigma_2}(A_S)
\parallel \left( \cdots \parallel \left( P_{\Sigma_{m-1}}(A_S)\parallel
P_{\Sigma_m}(A_S)\right)\right)\right) \nonumber
\end{eqnarray}
\end{lemma}
\begin{proof}
See the proof in the Appendix.
\end{proof}
\begin{lemma}\label{ParallelSimulation}
If two automata $A_2$ and $A_4$ (bi)simulate, respectively, $A_1$
and $A_3$, then $A_2\parallel \  A_4$ (bi)simulates $A_1\parallel \
A_3$, i.e.,
\begin{enumerate}
\item $\left(A_1\prec A_2 \right)\wedge \left(A_3\prec
A_4\right)\Rightarrow \left(A_1\parallel \ A_3\prec A_2\parallel \
A_4\right)$;
\item $\left(A_1\cong A_2 \right)\wedge \left(A_3\cong
A_4\right)\Rightarrow \left(A_1\parallel \ A_3\cong A_2\parallel \
A_4\right)$;
\end{enumerate}
 \end{lemma}
\begin{proof}
See the proof in the Appendix.
\end{proof}
\begin{theorem}\label{Top-down}
Consider a plant, represented by a deterministic parallel
distributed system $A_{\Delta}=\overset{n}{\underset{i=1}{\parallel}
}A_{P_i}$, with given local event sets $E_i$, $i=1,...,n$, and given
specification represented  by a deterministic decomposable automaton
$A_S \cong \overset{n}{\underset{i=1}{\parallel} }P_i(A_S)$, with
$E=\overset{n}{\underset{i=1}{\cup} }E_i$. If the Algorithm
\ref{Hierarchical Decomposition Algorithm 1} continues up to
$i=n-1$, then designing local controllers $A_{C_i}$, so that
$A_{C_i}\parallel A_{P_i}\cong P_i(A_S)$, $i=1,2,\cdots,n$, derives
the global closed loop system to satisfy the global specification
$A_S$, in the sense of bisimilarity, i.e.,
$\overset{n}{\underset{i=1}{\parallel} }(A_{C_i}\parallel
A_{P_i})\cong A_S$.
\end{theorem}
\begin{proof}
Algorithm \ref{Hierarchical Decomposition Algorithm 1} is a direct
extension of Theorem \ref{Task Automaton Decomposition} combined
with Lemma \ref{associativity-parallel}. Then, choosing local
controllers $A_{C_i}$, so that $A_{C_i}\parallel A_{P_i}\cong
P_i(A_S)$, $i=1,2,\cdots,n$, due to Lemma
$\ref{ParallelSimulation}.2$ leads to
$\overset{n}{\underset{i=1}{\parallel} }(A_{C_i}\parallel
A_{P_i})\cong \overset{n}{\underset{i=1}{\parallel} }P_i(A_S)\cong
A_S$.
\end{proof}

Now, if $DC1$-$DC4$ is reduced to $DC1$-$DC3$ (Theorem \ref{Task
Automaton Decomposition} is reduced into Lemma \ref{reverse
similarity relation of parallel}), then
$\overset{n}{\underset{i=1}{\parallel} }P_i(A_S)\cong A_S$ is
reduced into $\overset{n}{\underset{i=1}{\parallel} }P_i(A_S)\prec
A_S$, and hence, choosing local controllers $A_{C_i}$, so that
$A_{C_i}\parallel A_{P_i}\prec P_i(A_S)$, $i=1,2,\cdots,n$, due to
Lemma $\ref{ParallelSimulation}.1$ leads to
$\overset{n}{\underset{i=1}{\parallel} }(A_{C_i}\parallel
A_{P_i})\prec \overset{n}{\underset{i=1}{\parallel} }P_i(A_S)\prec
A_S$. Therefore,
\begin{corollary}
Considering the plant and global task as stated in Theorem
\ref{Top-down}, if $DC1$-$DC4$ is reduced to $DC1$-$DC3$ in
Algorithm \ref{Hierarchical Decomposition Algorithm 1} and it
continues up to $i=n-1$, then designing local controllers $A_{C_i}$,
so that $A_{C_i}\parallel A_{P_i}\prec P_i(A_S)$, $i=1,2,\cdots,n$,
derives the global closed loop system to satisfy the global
specification $A_S$, in the sense of similarity, i.e.,
$\overset{n}{\underset{i=1}{\parallel} }(A_{C_i}\parallel
A_{P_i})\prec A_S$.
\end{corollary}

\begin{remark}
It should be noted that in this approach, the parallel composition
requires a fixed communication pattern among the ``local'' automata
to synchronize on their common events. This framework is therefore
suitable for static distributed systems. For moving agent systems,
the agents are required to provide large enough communication range
to ensure that the connectivity is preserved during the movement, to
ensure the correct synchronization on the common events.
\end{remark}
\section{Example}\label{EXAMPLE}Consider a cooperative
multi-robot system (MRS) configured in Figure \ref{MRS-Cooperation}.
The  MRS consists of three robots $R_1$, $R_2$ and $R_3$. All robots
have the same communication and positioning capabilities.
Furthermore, the robot $R_2$ has the rescue and fire-fighting
capabilities, while $R_1$ and $R_3$ are normal robots with the
pushing capability. Initially, all of them are positioned in Room
$1$. Rooms $2$ and $3$ are accessible from Room $1$ by one-way door
$D_2$ and two-way doors $D_1$ and $D_3$, as shown in Figure
\ref{MRS-Cooperation}. All doors are equipped with spring to be
closed automatically, when there is no force to keep them open.

\begin{figure}[ihtp]
      \begin{center}
     \includegraphics[width=0.45\textwidth]{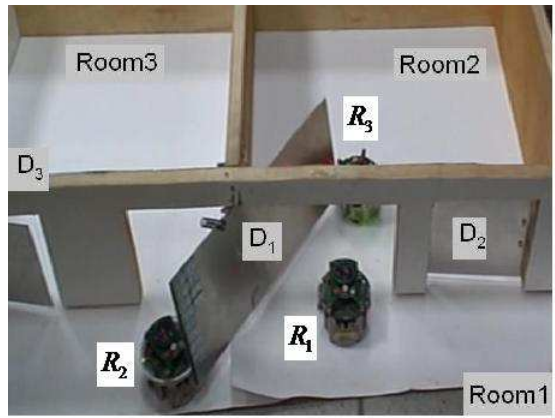}
        \caption{The environment of MRS coordination example.}
         \label{MRS-Cooperation}
        \end{center}
\end{figure}

Assume that Room $2$ requests help for fire extinguishing. After the
help announcement, the Robot $R_2$ is required to go to Room $2$,
urgently from $D_2$ and accomplish its task there and come back
immediately to Room $1$. However, $D_2$ is a one-way door, and,
$D_1$ is a heavy door and needs cooperation of two robots $R_1$ and
$R_3$ to be opened. To save time, as soon as the robots hear the
help request from Room $2$, $R_2$ and $R_3$ go to Rooms $2$ and $3$,
from $D_2$ and $D_3$, respectively, and then $R_1$ and $R_3$
position on $D_1$, synchronously open $D_1$ and wait for
accomplishment of task of $R_2$ in Room $2$ and returning to Room
$1$ ($R_2$ is fast enough). Afterwards, $R_1$ and $R_3$ move
backward to close $D_1$ and then $R_3$ returns back to Room $1$ from
$D_3$. All robots then stay at Room $1$ for the next
 task.

These requirements can be translated into a task automaton for the
robot team as it is illustrated in Figure $\ref{global task
automaton}$, defined over local event sets $E_1 = \{h_1, R_1toD_1,
R_1onD_1, FWD, D_1opened, R_2in1,\\ BWD, D_1closed, r\}$, $E_2=
\{h_2$, $R_2to2, R_2in2, D_1opened, R_2to1, R_2in1, r\}$, and $E_3 =
\{h_3, R_3to3,\\ R_3in3, R_3toD_1$, $R_3onD_1, FWD, D_1opened,
R_2in1, BWD, D_1closed, R_3to1, R_3in1, r\}$, with $h_i$:= $R_i$
received help request, $i= 1, 2, 3$; $R_jtoD_1$:=command for $R_j$
to position on $D_1$, $j= 1, 3$; $R_jonD_1$:= $R_j$ has positioned
on $D_1$, $j= 1,3$; $FWD$:= command for moving forward (to open
$D_1$); $BWD$:= command for moving backward (to close $D_1$) ;
$D_1opened$:= $D_1$ has been opened;  $D_1closed$:= $D_1$ has been
closed; $r$:= command to go to initial state for the next
implementation; $R_itok$:= command for $R_i$ to go to Room $k$, and
$R_iink$:= $R_i$ has gone to Room $k$, $i = 1, 2, 3$, $k = 1, 2, 3$.
\begin{figure}[ihtp]
      \begin{center}
     $A_S$:\\ \xymatrix@R=0.5cm{
 \bullet  \ar[r]^{r}  & \check{\bullet}\ar[r]^{h_1}\ar[d]_{h_2}\ar[dr]
 &\bullet \ar[r]^{R_1toD_1}\ar[dr] &\bullet \ar[r]^{R_1onD_1}\ar[dr] &\bullet \ar[dr]^{h_3}\\
\bullet \ar[u]^{R_3in1}& \bullet \ar[d]_{R_2to2}\ar[dr]&\bullet
\ar[d]\ar[dr]\ar[r]&\bullet \ar[r]\ar[dr] &\bullet
\ar[r]\ar[dr] &\bullet \ar[dr]^{R_3to3}\\
& \bullet \ar[d]_{R_2in2}\ar[dr] &\bullet \ar[d]\ar[dr]& \bullet
\ar[r]\ar[d]\ar[dr]
 &\bullet \ar[r]\ar[dr]
  &\bullet \ar[r]\ar[dr]
  &\bullet \ar[dr]^{R_3in3}\\
 \bullet
\ar[uu]^{R_3to1}&\bullet\ar[dr]&\bullet\ar[d]\ar[dr]
&\bullet\ar[d]\ar[dr] &\bullet\ar[d]\ar[dr]\ar[r]
&\bullet\ar[r]\ar[dr] &\bullet\ar[r]\ar[dr]
&\bullet\ar[dr]^{R_3toD_1}\\
&&\bullet\ar[dr]&\bullet\ar[d]\ar[dr] &\bullet\ar[d]\ar[dr]
&\bullet\ar[d]\ar[r]\ar[dr] &\bullet\ar[r]\ar[dr]
&\bullet\ar[r]\ar[dr]
&\bullet\ar[dr]^{R_3onD_1}\\
&&\bullet\ar[uull]^{D_1Closed}&\bullet\ar[dr] &\bullet\ar[d]\ar[dr]
&\bullet\ar[d]\ar[dr] &\bullet\ar[d]\ar[r] &\bullet\ar[d]\ar[r]
&\bullet\ar[d]\ar[r]
&\bullet\ar[d]\ar[dr]^{FWD}\\
&&&&\bullet\ar[dr] &\bullet\ar[d]\ar[dr] &\bullet\ar[d]\ar[r]
&\bullet\ar[d]\ar[r] &\bullet\ar[d]\ar[r]
&\bullet\ar[d]\ar[dr] &\bullet\ar[d]^{h_2}\\
&&&&\bullet\ar[uull]^{BWD} &\bullet\ar[dr] &\bullet\ar[d]\ar[r]
&\bullet\ar[d]\ar[r] &\bullet\ar[d]\ar[r]
&\bullet\ar[d]\ar[dr] &\bullet\ar[d]^{R_2to2}\\
 &&&&&&\bullet\ar[r]
&\bullet\ar[r]&\bullet\ar[r]&
\bullet\ar[dr]& \bullet\ar[d]^{R_2in2}\\
&&&&&&\bullet\ar[uull]^{R_2in1}
&&\bullet\ar[ll]^{R_2to1}&&\bullet\ar[ll]^{D_1opened}
                }
        \caption{Task automaton $A_S$ for robot team.}
\label{global task automaton}
        \end{center}
      \end{figure}
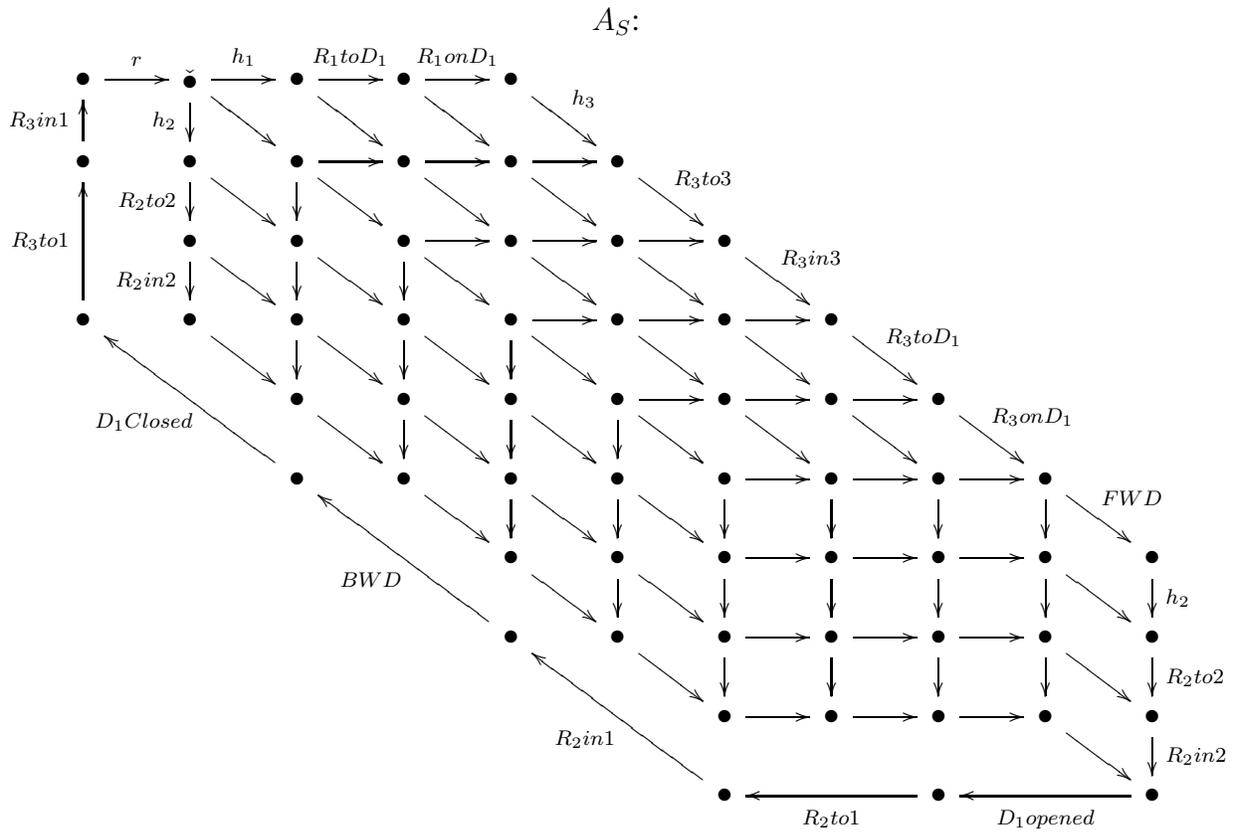

The states show the labels of transitions after  occurrence, and
hence are not labeled, here. Meaning that $\delta(x, e) = y$ is
interpreted as follows: Once $e$ occurs in $x$, it transits into $y$
and this state $y$ is labeled by $e (occurred)$.  For example, when
``help request'' event happens, the state of each robot transits
from the initial state to the state ``the robot received help
request''.

 We check decomposability condition for this global task
automaton with respect to $\Sigma_1 = E_2$ and $\bar{\Sigma}_1 =
E_1\cup E_3$ for the first stage in Algorithm $1$. Firstly, $\{h_2$,
$R_2to2, R_2in2\} = \Sigma_1\backslash \bar{\Sigma}_1$ can occur in
any order with respect to  $\{h_1, R_1toD_1, R_1onD_1, h_3, R_3to3,
R_3in3, R_3toD_1, R_3onD_1, FWD\}\\ = \bar{\Sigma}_1\backslash
\Sigma_1$, as it is shown in the global automaton in Figure
$\ref{global task automaton}$, satisfying $DC1$ and $DC2$. Moreover,
$D_1opened$, $R_2in1$ and $r$ are common events, provided by $R_1$,
$R_2$ and $R_3$, respectively, and informed to the other two robots
upon occurrence.

 Since
$\{D_1opened, FWD\}\subseteq \bar{\Sigma}_1$, $\{R_2in2,
D_1opened\}\subseteq \Sigma_1$, $\{D_1opened$, $R_2to1\}\subseteq
\Sigma_1$, $\{R_2to1, R_2in1\}\subseteq \Sigma_1$, $\{R_2in1,
BWD\}\subseteq \bar{\Sigma}_1$, $\{BWD, D_1closed\}\subseteq
\bar{\Sigma}_1$, $\{D_1closed,
 R_3to1\}\subseteq \bar{\Sigma}_1$, $\{R_3to1,
 R_3in1\}\subseteq \bar{\Sigma}_1$, $\{R_3in1,
 r\}\subseteq \bar{\Sigma}_1$, $\{r,
 h_1\}\subseteq \bar{\Sigma}_1$, $\{r,
 h_2\}\subseteq \Sigma_1$, $\{r,
 h_3\}\subseteq \bar{\Sigma}_1$, and hence all these successive transitions
 satisfy $DC1$ and $DC2$. Furthermore, all common events
 $\{D_1opened, R_2in1, r\} \subseteq \Sigma_1 \cap \bar{\Sigma}_1$ appear in
 only one branch, and hence, $DC3$ is satisfied. Finally,
 $P_{\Sigma_1}(A_S)$ and $P_{\bar{\Sigma}_1}(A_S)$ are both
 deterministic, and hence, $DC4$ is satisfied.

Therefore, due to Theorem \ref{Task Automaton Decomposition}, $A_S$
can be decomposed into $P_2(A_S) = P_{\Sigma_1}(A_S)=A_2$ and
$P_{1,3}=P_{\bar{\Sigma}_1}(A_S)$. The second stage of hierarchical
decomposition, decomposes $P_{1,3}(A_S)$ into $A_1=P_1(A_S)$ and
$A_3=P_3(A_S)$. The private transitions defined over $E_1\backslash
E_3=\{h_1, R_1toD_1, R_1onD_1\}$ can occur in any order with respect
to the transitions defined over the private local event set
$E_3\backslash E_1=\{h_3, R_3to3, R_3in3, R_3toD_1$, $R_3onD_1\}$.
Since $\{R_3onD_1, FWD\}\subseteq E_3$, $\{R_1onD_1, FWD\}\\
\subseteq E_1$, $\{FWD, D_1opened\}\subseteq E_1\cap E_3$,
$\{D_1opened, R_2in1\}\subseteq E_1 \cap E_3$, $\{R_2in1,
BWD\}\subseteq E_1 \cap E_3$, $\{BWD$, $D_1closed\}\subseteq E_1
\cap E_3$, $\{D_1closed, R_3to1 \}\subseteq E_3$, $\{R_3to1, R_3in1
\}\subseteq E_3$, $\{ R_3in1, r \}\subseteq E_3$, $\{r, h_3
\}\subseteq E_3$ and $\{r, h_1 \}\subseteq E_1$, then $DC1$ and
$DC2$ are satisfied. Furthermore, since all common events $\{FWD,
D_1opened, R_2in1, BWD$, $D_1closed, r\}\subseteq E_1 \cap E_3$
appear in only one branch in $P_{1,3}(A_S)$, therefore, there are no
pairs of strings violating $DC3$, and hence, $DC3$ is also
satisfied.  Moreover, $P_1(A_S)$ and $P_2(A_S)$ are both
deterministic, and consequently, $P_{1,3}(A_S)$ satisfies $DC4$. The
results of two decomposition stages are shown in Figures \ref{The
First Stage of Decomposition} and \ref{The Second Stage of
Decomposition}, such that $P_1(A_S)\parallel P_2(A_S)\parallel
P_3(A_S)\cong A_S$.

It can be seen that the design of supervisor to satisfy these
individual task automata is easier than the design of a global
supervisor to satisfy the global specification. Furthermore, since
the specification is determined for each agent, the global task can
be achieved in a decentralized fashion.

\begin{figure}[ihtp]
      $P_{1,3}(A_S)$: \begin{center}
      \xymatrix@R=0.5cm{
 \ar[r] & \bullet\ar[r]^{h_1}\ar[dr]
 &\bullet \ar[r]^{R_1toD_1}\ar[dr] &\bullet \ar[r]^{R_1onD_1}\ar[dr] &\bullet \ar[dr]^{h_3}\\
&\bullet \ar[u]^{r}& \bullet \ar[dr]\ar[r]&\bullet \ar[r]\ar[dr]
&\bullet
\ar[r]\ar[dr] &\bullet \ar[dr]^{R_3to3}\\
& & \bullet \ar[ul]^{R_3in1}
 &\bullet \ar[r]\ar[dr]
  &\bullet \ar[r]\ar[dr]
   &\bullet \ar[r]\ar[dr]
  &\bullet \ar[dr]^{R_3in3}\\
 &&&\bullet
\ar[ul]^{R_3to1}&\bullet\ar[dr]\ar[r] &\bullet\ar[r]\ar[dr]
&\bullet\ar[r]\ar[dr]
&\bullet\ar[dr]^{R_3toD_1}\\
&&&&\bullet\ar[ul]^{D_1closed} &\bullet\ar[r]\ar[dr]
&\bullet\ar[r]\ar[dr]&\bullet\ar[r]\ar[dr]
&\bullet\ar[dr]^{R_3onD_1}\\
&&&&&\bullet\ar[ul]^{BWD}& \bullet\ar[r] &\bullet\ar[r]
&\bullet\ar[r]
&\bullet\ar[dr]^{FWD}\\
&&&&&&\bullet\ar[ul]^{R_2in1} &&&&\bullet\ar[llll]^{D_1opened}
                }
 $P_2(A_S)$:  \xymatrix@C=0.5cm{
     \ar[r]&  \bullet \ar[r]^{h_2} &  \bullet  \ar[r]_{R_2to2}& \bullet \ar[r]^{R_2in2}& \bullet\ar[r]_{D_1opened}&
     \bullet\ar[r]^{R_2to1}& \bullet\ar[r]_{R_2in1}& \bullet
\ar`dr_l[llllll]`_u[llllll]_{r}[llllll]
   }
        \caption{$P_2(A_S)$ for $R_2$ and  $P_{\{E_1 \cup E_3\}}(A_S)$ for the team $\{R_1,R_3\}$.}
\label{The First Stage of Decomposition}
        \end{center}
      \end{figure}
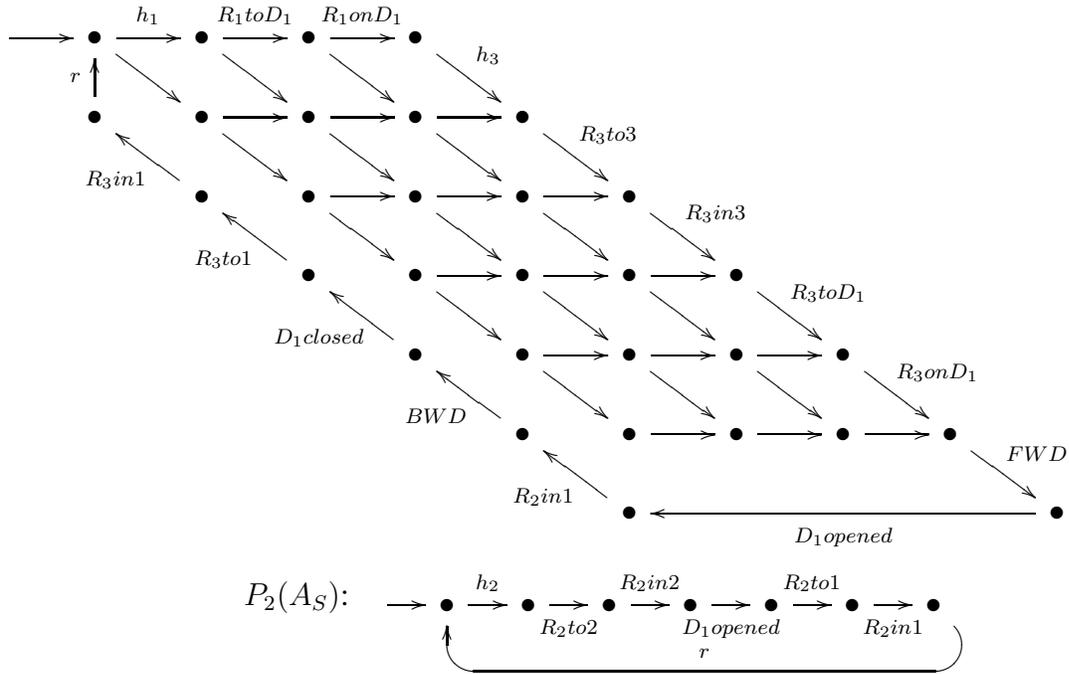

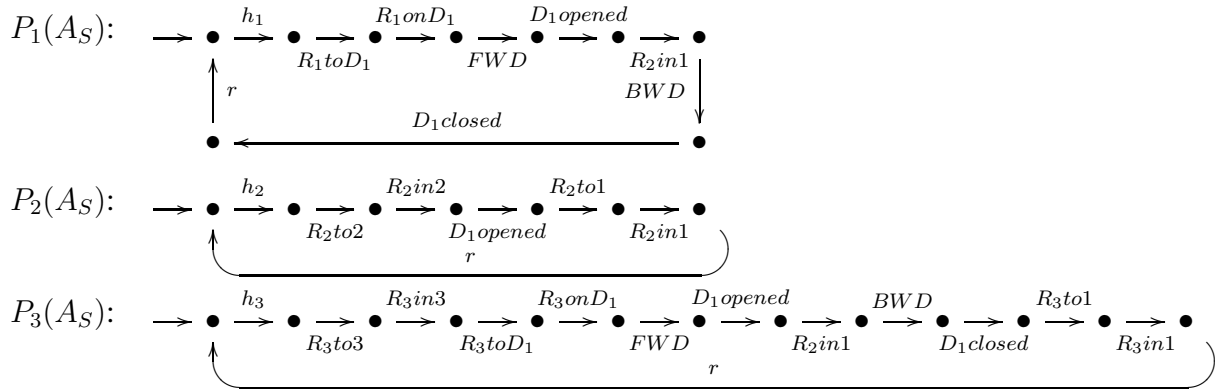
\begin{figure}[ihtp]
     $P_1(A_S)$:
      \xymatrix@C=0.5cm{
     \ar[r]&  \bullet \ar[r]^{h_1}&  \bullet\ar[r]_{R_1toD_1}& \bullet\ar[r]^{R_1onD_1}&
     \bullet\ar[r]_{FWD}& \bullet\ar[r]^{D_1opened}& \bullet\ar[r]_{R_2in1}&
     \bullet\ar[d]_{BWD}\\
     &\bullet \ar[u]_{r}&&&&&&\bullet\ar[llllll]_{D_1closed}
   }\\
$P_2(A_S)$: \xymatrix@C=0.5cm{
     \ar[r]&  \bullet \ar[r]^{h_2} &  \bullet  \ar[r]_{R_2to2}& \bullet \ar[r]^{R_2in2}& \bullet\ar[r]_{D_1opened}&
     \bullet\ar[r]^{R_2to1}& \bullet\ar[r]_{R_2in1}& \bullet
\ar`dr_l[llllll]`_u[llllll]_{r}[llllll]
   }\\
       $P_3(A_S)$:  \xymatrix@C=0.5cm{
     \ar[r]&  \bullet \ar[r]^{h_3}& \bullet \ar[r]_{R_3to3}& \bullet \ar[r]^{R_3in3}&
     \bullet\ar[r]_{R_3toD_1}& \bullet\ar[r]^{R_3onD_1}&
     \bullet\ar[r]_{FWD}& \bullet\ar[r]^{D_1opened}& \bullet\ar[r]_{R_2in1}& \bullet\ar[r]^{BWD}&
     \bullet\ar[r]_{D_1closed}&
     \bullet\ar[r]^{R_3to1}&
     \bullet\ar[r]_{R_3in1}&\bullet
\ar`dr_l[llllllllllll]`_u[llllllllllll]_{r}[llllllllllll]
   }
        \caption{$P_1(A_S)$ for $R_1$; $P_2(A_S)$ for $R_2$ and  $P_3(A_S)$ for $R_3$.}
 \label{The Second Stage of Decomposition}

      \end{figure}



Discussions on the design of supervisory control (for local local
plants and local task automata) and also refining of the low level
continuous controllers can be found in \cite{Cassandras2008},
\cite{Kumar1999}, \cite{Tabuada2003}, \cite{Belta2005} and
\cite{Belta2007}.

This scenario has been successfully implemented on a team of three
ground robots, shown in Figure \ref{MRS-Cooperation}.

\section{Conclusion}\label{CONCLUSION}

The paper proposed a formal method for automaton decomposition that
facilitates the top-down distributed cooperative control of
multi-agent systems. Given a set of agents whose logical behaviors
can be modeled as a parallel distributed system, and a global task
automaton, the paper has the following contributions: firstly, we
have provided necessary and sufficient conditions for
decomposability of an automaton with respect to parallel composition
and natural projections into two local event sets; secondly, the
approach has been extended into a sufficient condition for an
arbitrary finite number of agents, using a hierarchical algorithm,
and finally, we have shown that if a global task automaton is
decomposed into the local tasks for the individual agents, designing
the local supervisors for each agent, satisfying the local tasks,
guarantees that the closed loop system of the team of agents
satisfies the global specification.

The implementation of this approach is decentralized, in the sense
that there is no central unit to coordinate the agents; however,
they need communication to synchronize on common events. This
approach differs from the classical decentralized supervisory
control that refers to the configuration of a monolithic plant,
controlled by several supervisors that are distributed into
different nodes \cite{Cassandras2008}. The proposed approach is more
suitable for those applications with distributed configurations for
both plant and supervisor, such as multi-robot coordination systems
and sensor/actuator networks. In addition, due to associativity
property of parallel composition, the proposed approach can be
modular such that when a new task automaton is introduced to the
systems, one can decompose the new global task automaton, and then
compose the new local task automata with the corresponding old local
task automata.



\section{Appendix}

\subsection{Proof for Lemma \ref{similarity relation of parallel}}

We prove $A_S \prec P_1(A_S)||P_2(A_S))$ by showing that $R =
\{(q,z)\in Q\times Z | \exists s\in E^*, \delta(q_0, s) = q,
z=([q]_1, [q]_2) \}$ is a simulation relation, defined on all events
in $E$ and all reachable states in $A_S$. Consider $A_S = (Q, q_0,
E= E_1\cup E_2, \delta)$, $P_i(A_S) = (Q_i, q_i^0, E_i, \delta_i)$,
$i=1,2$, $P_1(A_S)||P_2(A_S) = (Z, z_0, E, \delta_{||})$. Then,
$\forall q, q^{\prime}\in Q, e\in E, \delta (q, e)= q^{\prime}$,
according to definition of natural projection (Definition
\ref{Natural Projection on Automaton}) $[q^{\prime}]_i = \left\{
  \begin{array}{ll}
\delta_i([q]_i, e) & \hbox{if $e\in E_i$;} \\
 {[q]}_i  & \hbox{if $e\notin E_i$}
  \end{array}
\right.$, $i = 1,2$, and due to definition of parallel composition
(Definition \ref{parallel composition}) $\delta_{||} (([q]_1, [q]_2)
, e) = ([q^{\prime}]_1, [q^{\prime}]_2) =
    \left\{
\begin{array}{ll}
    \left(\delta_1([q]_1, e), \delta_2([q]_2, e)\right), & \hbox{if $e\in E_1 \cap
    E_2$};\\
    \left(\delta_1([q]_1, e), [q]_2\right), & \hbox{if $e\in E_1 \backslash E_2$;} \\
    \left([q]_1, \delta_2([q]_2, e)\right), & \hbox{if $e\in E_2 \backslash E_1$.}
\end{array}\right.$

This is true for any $q\in Q$, particularly for $q_0$.  This
reasoning can be repeated for any reachable state in $Q$. Therefore,
starting from $q_0$ and taking $\left(q_0, Z_0= ([q_0]_1,
[q_0]_2)\right)\in R$, from the above construction, it follows that
for any reachable state in $Q$ ($\exists s\in E^*, \delta(q_0,
s)=q$) and $q^{\prime}\in Q$, $e\in E$, $\delta(q, e)= q^{\prime}$,
there exists $z = ([q]_1, [q]_2)$, $z^{\prime} = ([q^{\prime}]_1,
[q^{\prime}]_2)$ such that $\delta(z, e) = z^{\prime}$, and we can
take $(q, z)\in R$ and $(q^{\prime}, z^{\prime})\in R$. Therefore,
$R = \{(q,z)\in Q\times Z | \exists s\in E^*, \delta(q_0, s)!,
z=([q]_1, [q]_2) \}$ is a simulation relation, defined over all
$e\in E$, and all reachable states in $A_S$, and hence, $A_S \prec
P_1(A_S)||P_2(A_S)$.
\subsection{Proof for Lemma \ref{reverse similarity relation of parallel}}
We use two following lemmas during the proof.
\begin{lemma}\label{Gamma1}
Consider a deterministic automaton $A_S=(Q, q_0, E = E_1\cup E_2,
\delta)$. Then $DC1 \wedge DC2 \Rightarrow [\forall s\in E^*$,
$\delta(q_0, s)!\Rightarrow \delta(q_0, p_1(s)|p_2(s))!$ in $A_S$.
\end{lemma}
This lemma means that for any transition defined on a string in
$A_S$, all path automata defined on the interleaving of $p_1(s)$ and
$p_2(s)$ in $P_1(A_S)||P_2(A_S)$ are simulated by $A_S$, provided
$DC1$ and $DC2$.

\begin{proof}
Consider a deterministic automaton $A_S=(Q, q_0, E = E_1\cup E_2,
\delta)$, a string $s\in E^*$, $\delta(q_0, s)=q$, and its
projections $p_1(s)$, $p_2(s)$ with $\delta_1(x_0, p_1(s))$ = $x$,
$\delta_2(y_0, p_2(s))= y$ and $(x,y) \in \delta_{||}((x_0, y_0),
p_1(s)|p_2(s))$, in $P_1(A_S)$, $P_2(A_S)$ and $P_1(A_S)||P_2(A_S)$,
respectively. Any string $s$ can be written as $s=\omega^1
\gamma^1...\ \omega^K \gamma^K$, with $\omega^k\in
[E\backslash(E_1\cap E_2)]^*$, $\gamma^k \in (E_1\cap E_2)^*$,
$p_1(\omega^k)=\alpha^k = \alpha^k_0...\alpha^k_{m_k}$, $\alpha^k_0
= \varepsilon$, $p_2(\omega^k)=\beta^k = \beta^k_0...\beta^k_{n_k}$,
$\beta^k_0 = \varepsilon$, $p_1(\gamma^k) = p_2(\gamma^k)=\gamma^k =
\gamma^k_0...\gamma^k_{r_k}$, $\gamma^k_0 = \varepsilon$. The case
$m_k=0$, $n_k=0$, $r_k=0$, $K=0$, results in
$p_1(\omega^k)=\varepsilon$, $p_2(\omega^k)=\varepsilon$, $\gamma
^k=\varepsilon$ and $s=\varepsilon$. Based on this setting and
definition of parallel composition, for $k=0,...,K$, $i=0,...,m_k$,
$j=0,...,n_k$ and $r=0,...,r_k$, the interleaving $p_1(s)|p_2(s)$ is
evolved in $P_1(A_S)||P_2(A_S)$ as follows: $\forall
(x_{k+i},y_{k+j}) \in Q_1\times Q_2$:
$\delta_{||}((x_{k+i},y_{k+j}), \alpha_i^k)$=$ (\delta_1(x_{k+i},
\alpha_i^k), y_{k+j})$, $\delta_{||}((\delta_1(x_{k+i}, \alpha_i^k),
y_{k+j}), \beta_j^k)$=$ ((\delta_1(x_{k+i}, \alpha_i^k),
\delta_2(y_{k+j}, \beta_j^k))$, $\delta_{||}((x_{k+i}, y_{k+j}),
\beta_j^k)$=$ (x_{k+i}, \delta_2(y_{k+j}, \beta_j^k))$,
$\delta_{||}((x_{k+i}, \delta_2(y_{k+j}, \beta_j^k)), \alpha_i^k)
$=$(\delta_1(x_{k+i}, \alpha_i^k), \delta_2(y_{k+j}, \beta_j^k))$,
$\delta_{||}((x_{k+m_k+r}, y_{k+n_k+r}), \gamma_r^k))=\\
(\delta_1(x_{k+m_k+r}, \gamma_r^k), \delta_2(y_{k+n_k+r},
\gamma_r^k))$ with $\delta_1(x_{k+i}, \alpha_i^k) =
\left\{\begin{array}{ll}
x_{k+i} & \mbox{if } \alpha^k_i=\varepsilon \\
x_{k+i+1} & \mbox{if } \alpha^k_i\neq\varepsilon
\end{array}\right.$, $\delta_2(y_{k+j}, \beta_j^k) =
\left\{\begin{array}{ll}
y_{k+j} & \mbox{if } \beta^k_j=\varepsilon \\
y_{k+j+1} & \mbox{if } \beta^k_j\neq\varepsilon
\end{array}\right.$ and $(\delta_1(x_{k+m_k+r}, \gamma_r^k), \delta_2(y_{k+n_k+r},
\gamma_r^k))=\\ \left\{\begin{array}{ll}
(x_{k+m_k+r}, y_{k+n_k+r}) & \mbox{if} \gamma^k_r=\varepsilon \\
(x_{k+m_k+r+1}, y_{k+n_k+r+1})  & \mbox{if
}\gamma^k_r\neq\varepsilon
\end{array}\right.$.

Moreover, $DC1$ and $DC2$ collectively imply that  $\forall e_1 \in
E_1\backslash E_2$, $e_2 \in E_2\backslash E_1$, $q\in Q$,
$[\delta(q,e_1)!\wedge \delta(q,e_2)!]\vee \delta(q, e_1e_2)!\vee
\delta(q, e_2e_1)!\Rightarrow \delta(q, e_1e_2)! \wedge \delta(q,
e_2e_1)!$ which particularly means that $\forall k\in\{0,...K\}$,
$\forall \alpha^k_i, \beta^k_j, q_{i,j}\in Q$, $i = 0,...,m_k$,
$j=0,...,n_k$, $r=0,...,r_k$: $\delta(q_{k+i,k+j},
\alpha^k_i\beta^k_j)! \wedge \delta(q_{k+i,k+j},
\beta^k_j\alpha^k_i)!$ with a simulation relation $R(\omega^k)=\{
((x_{k+i},y_{k+j}),q_{k+i,k+j})$,
 $(\delta_1(x_{k+i}, \alpha_i^k), y_{k+j})$, $\delta(q_{k+i,k+j}, \alpha_i^k))$,
$((x_{k+i},\delta_2(y_{k+j}, \beta_j^k)), \delta(q_{k+i,k+j},
\beta_j^k))$, $((\delta_1(x_{k+i}, \alpha_i^k),
\delta_2(y_{k+j},\beta_j^k))$, $\delta(q_{k+i,k+j},\\
\alpha_i^k\beta_j^k))$, $((\delta_1(x_{k+i}, \alpha_i^k),
\delta_2(y_{k+j}, \beta_j^k)), \delta(q_{k+i,k+j},
\beta_j^k\alpha_i^k))\}$
 from transitions defined on $p_1(\omega^k)|p_2(\omega^k)$ into $A_S$. For the transitions on
the common events, the evolutions are $\delta(q_{k+m_k+r, k+n_k+r},
\gamma_r^k)!$ in $A_S$ for $r = 0,..., r_k$, leading to simulation
relation $R(\gamma^k)= \{((x_{k+m_k+r}, y_{k+n_k+r})$,$ q_{k+m_k+r,
k+n_k+r})$, $((\delta_1(x_{k+m_k+r}, \gamma_r^k)$,$
\delta_2(y_{k+n_k+r}, \gamma_r^k))$, $\delta(q_{k+m_k+r, k+n_k+r}$,
$\gamma_r^k))\}$ from transitions on $p_1(\gamma^k)|p_2(\gamma^k)$
into $A_S$. Therefore, for $i=0,...,m_k$, $j=0,...,n_k$, $r =
0,...,r_k$, $R= \overset{K}{\underset{k=0}{\cup\ }} R(\omega^k)\cup
R(\gamma^k)$ defines a simulation relation from $PA(z_0,
p_1(s)|p_2(s))$ (in $P_1(A_S)||P_2(A_S)$) into $A_S$.
\end{proof}
\begin{lemma}\label{Gamma2}
If $DC1$ and $DC2$ hold true, then $\forall s, s^{\prime} \in E^*$,
$\delta(q_0, s)!$, $\delta(q_0, s^{\prime})!$, $s\neq s^{\prime}$,
$p_{E_1\cap E_2}(s)$, $p_{E_1\cap E_2}(s^{\prime})$ do not start
with the same $a\in E_1\cap E_2$, then $\delta(q_0,
p_1(s)|p_2(s^{\prime}))! \wedge \delta(q_0,
p_1(s^{\prime})|p_2(s))!$ in $A_S$.
\end{lemma}
\begin{proof}
The antecedent of Lemma \ref{Gamma2} addresses three following
cases: (Case $1$): $s=\omega_1$, $s^{\prime} = \omega_1^{\prime}$;
(Case $2$): $s=\omega_1a\omega_2$, $s^{\prime} = \omega_1^{\prime}$,
and (Case $3$): $s=\omega_1a\omega_2$, $s^{\prime} =
\omega_1^{\prime}b\omega_2^{\prime}$, where, $\omega_1,
\omega_1^{\prime}\in [E\backslash (E_1\cap E_2)]^*$, $\omega_2,
\omega_2^{\prime}\in (E_1\cup E_2)^*$, $a,b \in E_1\cap E_2$.

For Case $1$, setting $K=1$, $r_k=0$, taking $p_1(\omega_1)= \alpha
= \alpha_0...\alpha_m \in (E_1\backslash E_2)^*$,
$p_2(\omega_1^{\prime}) = \beta^{\prime} =
\beta^{\prime}_0...\beta^{\prime}_{n^{\prime}}\in (E_2\backslash
E_1)^*$, $p_1(\omega_1^{\prime}) = \alpha^{\prime} =
\alpha_0^{\prime}...\alpha^{\prime}_{m^{\prime}} \in (E_1\backslash
E_2)^*$ and $p_2(\omega_1) = \beta = \beta_0...\beta_n \in
(E_2\backslash E_1)^*$, similar to the first part of Lemma
\ref{Gamma1}, it follows that $\delta(q_0,
p_1(\omega_1)|p_2(\omega^{\prime}_1))!$ and\\ $\delta(q_0,
p_1(\omega^{\prime}_1)|p_2(\omega_1))!$ in $A_S$.

For Case $2$, from Case $1$ and Lemma \ref{Gamma1} it follows that
$\delta(q_0, p_1(s)|p_2(s^{\prime}))!=
\delta(q_0,[p_1(\omega_1)|p_2(\omega^{\prime}_1)]\\
ap_1(\omega_2))!$ and $\delta(q_0, p_1(s^{\prime})|p_2(s))!=
\delta(q_0, [p_1(\omega^{\prime}_1)|p_2(\omega_1)]a p_2(\omega_2))!$
in $A_S$.

For the third case, from definition of parallel composition combined
with the first two cases and also Lemma \ref{Gamma1}, $\delta(q_0,
p_1(s)|p_2(s^{\prime}))!$ leads to
$\delta(q_0,[p_1(\omega_1)|p_2(\omega^{\prime}_1)]
a[p_1(\omega_2)|p_2(\omega_2)])!$ and
$\delta(q_0,[p_1(\omega_1)|\\p_2(\omega^{\prime}_1)]
b[p_1(\omega_2^{\prime})|p_2(\omega_2^{\prime})])!$ in $A_S$ and
similarly, $\delta(q_0, p_1(s^{\prime})|p_2(s))!$ results in
$\delta(q_0,
[p_1(\omega^{\prime}_1)|p_2(\omega_1)]\\a[p_1(\omega_2)|p_2(\omega_2)])!$
and $\delta(q_0,
[p_1(\omega^{\prime}_1)|p_2(\omega_1)]b[p_1(\omega_2^{\prime})|
p_2(\omega_2^{\prime}))])!$ in $A_S$.

Therefore, for all three cases, $\delta(q_0,
p_1(s)|p_2(s^{\prime}))!$ in $A_S$ and $\delta(q_0, p_1(s^{\prime})|
p_2(s))!$ in $A_S$.
\end{proof}

Now, Lemma \ref{reverse similarity relation of parallel} is proven
as follows.\\ \textbf{Sufficiency:} The set of transitions in
$P_1(A_S)||P_2(A_S)$, is defined as
$T=\{(x_0,y_0)\overset{p_1(s)|p_2(s^{\prime})}\longrightarrow
(x,y)\in Q_1 \times Q_2\}$, where,
$(x_0,y_0)\overset{p_1(s)|p_2(s^{\prime})}\longrightarrow (x,y)$ in
$P_1(A_S)||P_2(A_S)$ is the interleaving of transitions
$x_0\overset{p_1(s)}\longrightarrow x$ in $P_1(A_S)$ and a
transition $y_0\overset{p_2(s^{\prime})}\longrightarrow y$ in
$P_2(A_S)$ (projections of transitions
$q_0\overset{s}\longrightarrow q$ and
$q_0\overset{s^{\prime}}\longrightarrow q^{\prime}$, respectively,
in $A_S$). $T$ can be divided into three sets of transitions
corresponding to a division $\{\Gamma_1, \Gamma_2, \Gamma_3\}$ on
the set of interleaving strings
$\Gamma=\{p_1(s)|p_2(s^{\prime})|q_0\overset{s}\longrightarrow q,
q_0\overset{s^{\prime}}\longrightarrow q^{\prime}, q,q^{\prime}\in
Q, s, s^{\prime}\in E^* \}$, where,
$\Gamma_1=\{p_1(s)|p_2(s^{\prime})\in \Gamma |s=s^{\prime}\}$,
$\Gamma_2=\{p_1(s)|p_2(s^{\prime})\in \Gamma |s\neq s^{\prime},
p_{E_1\cap E_2}(s)$ and $p_{E_1\cap E_2}(s^{\prime})$ do not start
with the same event$\}$, and $\Gamma_3=\{p_1(s)|p_2(s^{\prime})\in
\Gamma |s\neq s^{\prime}, p_{E_1\cap E_2}(s)$ and $p_{E_1\cap
E_2}(s^{\prime})$ start with the same event $\}$.

Now, for any $s$, $s^{\prime}$, $\delta(q_0, s)!$, $\delta(q_0,
s^{\prime})!$, in $A_S$, both $\delta(q_0, p_1(s)|p_2(s^{\prime}))!$
and $\delta(q_0, p_1(s^{\prime})|p_2(s))!$ are guaranteed, for
$\Gamma_1$, due to Lemma \ref{Gamma1}; for $\Gamma_2$, due to Lemma
\ref{Gamma2}, and for $\Gamma_3$, due to combination of $DC3$ and
Lemma \ref{Gamma1} (For simplification, in $DC3$, $s$ and
$s^{\prime}$ can be started from any $q$, instead of $q_0$, and the
strings between $q_0$ and $q$ are checked by Lemma \ref{Gamma1}).

\textbf{Necessity:} The necessity is proven by contradiction.
Suppose that $A_S$ simulates $P_1(A_S)||P_2(A_S)$, but $\exists
e_1\in E_1\backslash E_2, e_2 \in E_2\backslash E_1, q\in Q, s\in
E^*$ s.t. $(1)$: $[\delta(q,e_1)!\wedge \delta(q,e_2)!]\wedge \neg
[\delta(q, e_1e_2)!\wedge \delta(q, e_2e_1)!]$; $(2)$:
$\neg[\delta(q, e_1e_2s)!\Leftrightarrow \delta(q, e_2e_1s)!]$, or
$(3)$: $\exists s, s^{\prime} \in E^*$, sharing the same first
appearing common event $a\in E_1 \cap E_2$, $s\neq s^{\prime}$,
$q\in Q$: $\delta(q, s)! \wedge \delta(q, s^{\prime})! \wedge
\neg[\delta(q, p_1(s)|p_2(s^{\prime}))! \wedge \delta(q,
p_1(s^{\prime})|p_2(s))!]$.

In the first case, due to definition of parallel composition, the
expression $[\delta(q,e_1)!\wedge \delta(q,e_2)!]$ leads to
$\delta_{||}(z, e_1e_2)!=\delta_{||}(z, e_2e_1)!$, where $z\in Q_1
\times Q_2$ in $P_1(A_S)||P_2(A_S)$ corresponds to $q\in Q$ in
$A_S$. Therefore, $\delta_{||}(z, e_1e_2)!\wedge\delta_{||}(z,
e_2e_1)!$, but, $\neg [\delta(q, e_1e_2)!\wedge \delta(q,
e_2e_1)!]$. This means that $P_1(A_S)||P_2(A_S)\nprec A_S$ which is
a contradiction. The second case means $\exists e_1\in E_1\backslash
E_2, e_2 \in E_2\backslash E_1, q\in Q, s\in E^*$ s.t. $[\delta(q,
e_1e_2s)! \vee \delta(q, e_2e_1s)!]\wedge \neg [\delta(q, e_1e_2s)!
\wedge \delta(q, e_2e_1s)!]$. From definition of parallel
composition, then $\delta(q, e_1e_2s)! \vee \delta(q, e_2e_1s)!$
implies that $\delta_{||}(z, e_1e_2)! =\delta_{||}(z, e_2e_1)!$, for
some $z\in Q_1 \times Q_2$ corresponding to $q\in Q$. Consequently,
from definition of transition relation and $A_S\prec
P_1(A_S)||P_2(A_S)$ it turns to $\delta_{||}(z, e_1e_2s)!
=\delta_{||}(z, e_2e_1s)!$, meaning that
 $\delta_{||}(z, e_1e_2s)! \wedge
\delta_{||}(z, e_2e_1s)!$, but, $\neg [\delta(q, e_1e_2s)! \wedge
\delta(q, e_2e_1s)!]$. This in turn contradicts with the similarity
assumption of $P_1(A_S)||P_2(A_S)\prec A_S$. The third case also
leads to the contradiction as causes the violation of the simulation
relation from $P_1(A_S)||P_2(A_S)$ into $A_S$ as $\delta(q, s)!
\wedge \delta(q, s^{\prime})!$ leads to $\delta_{||}(z,
p_1(s)|p_2(s^{\prime}))! \wedge \delta_{||}(z,
p_2(s)|p_1(s^{\prime}))!$ in $P_1(A_S)||P_2(A_S)$, whereas
$\neg[\delta(q, p_1(s)|p_2(s^{\prime}))! \wedge \delta(q,
p_2(s)|p_1(s^{\prime}))!]$.
\subsection{Proof for Lemma \ref{symmetric}}
To prove  Lemma \ref{symmetric}, we use following two lemmas
together with Lemma \ref{ParallelSimulation}. Firstly, the following
lemma is introduced to characterize the symmetric property of
simulation relations.
\begin{lemma}\label{symmetric simulations and determinism}
Consider two automata $A_1$ and $A_2$, and let $A_1$ be
deterministic, $A_1\prec A_2$ with the simulation relation $R_1$ and
$A_2\prec A_1$ with the simulation relation $R_2$. Then, $R_1^{-1} =
R_2$ if and only if there exists a deterministic automaton
$A^{\prime}_1$ such that $A_1^{\prime}\cong A_2$.
\end{lemma}
\begin{proof}

\textbf{Sufficiency:} $A_1\prec A_2$, $A_2\prec A_1$ and
$A_1^{\prime}\cong A_2$, collectively, result in $A_1\prec
A_1^{\prime}$ and $A_1^{\prime}\prec A_1$, that due to determinism
of $A_1$ and $A_1^{\prime}$ lead to $A_1 \cong A_1^{\prime}$.
Finally, since $A_1^{\prime}\cong A_2$, from transitivity of
bisimulation, $A_1 \cong A_2$, and consequently, $R_1^{-1} = R_2$.

\textbf{Necessity:} The necessity is proven by contradiction as
follows. Consider two automata $A_1 = (X, x_0, E, \delta_1)$, $A_2 =
(Y, y_0, E, \delta_2)$, let $A_1$ be deterministic, $A_1\prec A_2$
with the simulation relation $R_1$, $A_2\prec A_1$ with the
simulation relation $R_2$ and suppose that $R_1^{-1} = R_2$ (and
hence $A_1\cong A_2$), but there does not exist a deterministic
automaton $A^{\prime}_1$ such that $A^{\prime}_1\cong A_2$. This
means that $\exists s\in E^*$, $\sigma\in E$, $y_1, y_2 \in Y$,
$\delta_2 (y_0, s) = y_1$, $\delta_2 (y_0, s) = y_2$, $\delta_2
(y_1, \sigma)!$, but $\neg\delta_2 (y_2, \sigma)!$. From $A_2 \prec
A_1$, $\delta_2 (y_0, s) = y_1 \wedge \delta_2 (y_1, \sigma)!$
implies that $\exists x_1\in X$, $\delta_1 (x_0, s) =x_1 \wedge
\delta_1(x_1, \sigma)!$. On the other hand, $A_1$ is deterministic,
and hence, $\forall x_2\in X$, $\delta_1(x_0, s) = x_2 \Rightarrow
x_2 = x_1$. Therefore, $A_2\prec A_1$ necessarily leads to $(y_2,
x_1)\in R_2$. But, $\exists \sigma \in E$ such that $\delta_1(x_1,
\sigma)! \wedge \neg\delta_2(y_2, \sigma)!$,  meaning that $(y_2,
x_1) \in R_2 \wedge \neg (x_1, y_2) \notin R_1$, i.e., $R_1^{-1}
\neq R_2$, that contradicts  with the hypothesis, and the necessity
is followed.
\end{proof}
Next, let $A_1$ and $A_2$ to be substituted by $A_S$ and
$P_1(A_S)||P_2(A_S)$, respectively, in Lemma \ref{symmetric
simulations and determinism}. Then, the existence of $A_1^{\prime} =
A_S^{\prime}$ in Lemma \ref{symmetric simulations and determinism}
is characterized by the following lemma.
\begin{lemma}\label{decomposition and determinism}
Consider a deterministic automaton $A_S$ and its natural projections
$P_i(A_S)$, $i = 1, 2$. Then, there exists a deterministic automaton
$A^{\prime}_S$ such that $A^{\prime}_S\cong P_1(A_S)||P_2(A_S)$ if
and only if there exist deterministic automata $P_i^{\prime}(A_S)$
such that $P_i^{\prime}(A_S) \cong P_i(A_S)$, $i = 1, 2$.
\end{lemma}
\begin{proof}
Let $A_S = (Q, q_0, E = E_1\cup E_2, \delta)$, $P_i(A_S) = (Q_i,
q_0^i, E_i, \delta_i)$, $P_i^{\prime}(A_S) = (Q_i^{\prime},
q_{0,i}^{\prime}, E_i, \delta_i^{\prime})$, $i = 1, 2$,
$P_1(A_S)||P_2(A_S) = (Z, z_0, E, \delta_{||})$,
$P_1^{\prime}(A_S)||P_2^{\prime}(A_S) = (Z^{\prime}, z_0^{\prime},
E, \delta_{||}^{\prime})$. The proof of Lemma \ref{decomposition and
determinism} is then presented as follows.

\textbf{Sufficiency:} The existence of deterministic automata
$P_i^{\prime}(A_S)$ such that $P_i^{\prime}(A_S) \cong P_i(A_S)$, $i
= 1, 2$ implies that $\delta^{\prime}_1$ and $\delta^{\prime}_2$ are
functions, and consequently from definition of parallel composition
(Definition \ref{parallel composition}), $\delta_{||}^{\prime}$ is a
function, and hence $P_1^{\prime}(A_S)||P_2^{\prime}(A_S)$ is
deterministic. Moreover, from Lemma \ref{ParallelSimulation},
$P_i^{\prime}(A_S) \cong P_i(A_S)$, $i = 1, 2$ leads to
$P_1^{\prime}(A_S)||P_2^{\prime}(A_S)\cong P_1(A_S)||P_2(A_S)$,
meaning that there exists a deterministic automaton $A^{\prime}_S =
P_1^{\prime}(A_S)||P_2^{\prime}(A_S)$ such that $A^{\prime}_S \cong
P_1(A_S)||P_2(A_S)$.

\textbf{Necessity:} The necessity is proven by contraposition,
namely, by showing that if there does not exist deterministic
automata $P_i^{\prime}(A_S)$ such that $P_i^{\prime}(A_S) \cong
P_i(A_S)$, for $i = 1$ or $i = 2$, then there does not exist a
deterministic automaton $A^{\prime}_S$ such that $A^{\prime}_S \cong
P_1(A_S)||P_2(A_S)$.

Without loss of generality, assume that there does not exist a
deterministic automaton $P_1^{\prime}(A_S)$ such that
$P_1^{\prime}(A_S) \cong P_1(A_S)$. This means that $\exists q, q_1,
q_2 \in Q$, $e\in E_1$, $t_2\in (E_2\backslash E_1)^*$, $t\in E^*$,
$\delta(q, t_2e) = q_1$, $\delta(q, e) = q_2$, $\neg(\delta(q_1, t)
= q_2\Leftrightarrow \delta(q_2, t)!)$, meaning that $\delta(q_1,
t)!\wedge \neg\delta(q_2, t)!$ or $\neg\delta(q_1, t)!\wedge
\delta(q_2, t)!$. Again without loss of generality we consider the
first case and show that it leads to a contradiction. From the first
case, $\delta(q_1, t)!\wedge \neg\delta(q_2, t)!$, and definition of
natural projection, it follows that $\delta_1([q]_1, e) = [q_1]_1$,
$\delta_1([q_1]_1, p_1(t))!$, $\delta_1([q]_1, e) = [q_2]_1$,
$\neg\delta_1([q_2]_1, p_1(t))!$, $\delta_2([q]_2, p_2(e)) =
[q_2]_2$, $\neg\delta_2([q_2]_2, p_2(t))!$, and hence,
$\delta_{||}(([q]_1, [q]_2), e) = ([q_1]_1, [q_1]_2)$,
$\delta_{||}(([q_1]_1, [q_1]_2),\\ p_1(t))!$, whereas
$\delta_{||}(([q]_1, [q]_2), e) = ([q_1]_1, [q_2]_2)$,
$\neg\delta_{||}(([q_1]_1, [q_2]_2), p_1(t))!$ in
$P_1(A_S)||P_2(A_S)$, implying that there does not exist a
deterministic automaton $A^{\prime}_S$ such that $A^{\prime}_S \cong
P_1(A_S)||P_2(A_S)$, and the necessity is proven.
\end{proof}

Now, Lemma \ref{symmetric} is proven as follows.

\textbf{Sufficiency:} $DC4$ implies that there exist deterministic
automata $P_i^{\prime}(A_S)$ such that $P_i^{\prime}(A_S) \cong
P_i(A_S)$, $i = 1, 2$. Then, from Lemmas \ref{ParallelSimulation}
and \ref{decomposition and determinism}, it follows, respectively,
that $P_1^{\prime}(A_S)||P_2^{\prime}(A_S)\cong P_1(A_S)||P_2(A_S)$,
and that there exists a deterministic automaton $A^{\prime}_S =
P_1^{\prime}(A_S)||P_2^{\prime}(A_S)$ such that $A^{\prime}_S \cong
P_1(A_S)||P_2(A_S)$ that due to Lemma \ref{symmetric simulations and
determinism}, it results in $R_1^{-1} = R_2$.

\textbf{Necessity:} Let $A_S$ be deterministic, $A_S\prec
P_1(A_S)||P_2(A_S)$ with the simulation relation $R_1$ and
$P_1(A_S)||P_2(A_S)\prec A_S$ with the simulation relation $R_2$,
and assume by contradiction that $R_1^{-1} = R_2$, but $DC4$ is not
satisfied. Violation of $DC4$ implies that for $i = 1$ or $i = 2$,
there does not exists a deterministic automaton $P_i^{\prime}(A_S)$
such that $P_i^{\prime}(A_S) \cong P_i(A_S)$. Therefore, due to
Lemma \ref{decomposition and determinism}, there does not exist a
deterministic automaton $A^{\prime}_S$ such that $A^{\prime}_S \cong
P_1(A_S)||P_2(A_S)$, and hence, according to Lemma \ref{symmetric
simulations and determinism}, it leads to $R_1^{-1} \neq R_2$ which
is a contradiction.

\subsection{Proof for Lemma \ref{associativity-parallel}}
Lemma \ref{associativity-parallel} comes from the associativity
property of parallel decomposition \cite{Cassandras2008} as for
automata $A_i$, $i = 1,...,m$: $A_ 1\parallel A_ 2 \parallel \cdots
\parallel
A_{m-1}\parallel A_{E_m} = A_1\parallel \left(A_2
\parallel \left( \cdots \parallel \left( A_{m-1}\parallel
A_m\right)\right)\right)$.
\subsection{Proof for Lemma \ref{ParallelSimulation}}
Lemma \ref{ParallelSimulation}.1 is proven by showing that the
relation $R = \{ ((q_1,q_3), (q_2,q_4)) | (q_1,q_2) \in R_1 \hbox{
and }\\ (q_3,q_4) \in R_2 \}$ is a simulation relation, where, $R_1$
and $R_2$ are the respective simulations from $A_1$ to $A_2$ and
from $A_3$ to $A_4$.

 Consider $A_i=(Q_i, q_i^0, E_i, \delta_i)$, $i=1,...,4$,
 $A_1||A_3 = (Q_{1,3}, (q_1^0, q_3^0), E = E_1\cup E_3,
 \delta_{1,3})$,  $A_2||A_4 = (Q_{2,4}, (q_2^0, q_4^0), E = E_2\cup E_4,
 \delta_{2,4})$, $E_1 = E_3$ and $E_2 = E_4$. Then, $\forall (q_1, q_3), (q_1, q_3)^{\prime}\in Q_{1,3}$, $e\in
 E$, $q_2\in Q_2$, $q_4\in Q_4$ such that
 $\delta_{1,3}((q_1, q_3), e) = (q_1, q_3)^{\prime}$, $(q_1, q_2)\in R_1$ and $(q_3, q_4)\in
 R_2$, according to definition of parallel composition (Definition
\ref{parallel composition}), we have $ (q_1, q_3)^{\prime} =
(q_1^{\prime}, q_3^{\prime}) = \left\{
\begin{array}{ll}
    \left(\delta_1(q_1, e), \delta_3(q_3, e)\right), & \hbox{if $\delta_1(q_1, e)!$, $\delta_3(q_3, e)!$, $
     e\in E_1 \cap E_3$};\\
    \left(\delta_1(q_1, e), q_3\right), & \hbox{if $\delta_1(q_1, e)!$, $e\in E_1 \backslash E_3$;} \\
    \left(q_1, \delta_3(q_3, e)\right), & \hbox{if $\delta_3(q_3, e)!$, $e\in E_3 \backslash E_1$;}
\end{array}\right.$\\ and due to definition of simulation (Definition
 \ref{simulation}), $A_1 \prec A_2$  and $A_3 \prec A_4$, it follows
 that
 $\left\{
\begin{array}{ll}
    \exists q_i^{\prime}\in Q_i, \delta_i(q_i,
    e)= q_i^{\prime}, i = 2, 4& \hbox{ if $e\in E_2\cap E_4$}\\
    \exists q_2^{\prime}\in Q_2, \delta_2(q_2,
    e)= q_2^{\prime}& \hbox{ if $e\in E_2\backslash E_4$}\\
        \exists q_4^{\prime}\in Q_4,  \delta_4(q_4,
    e)= q_4^{\prime}& \hbox{ if $e\in E_4\backslash E_2$}
\end{array}\right.$

This, in turn, due to definition of parallel composition implies
that $\exists (q_2^{\prime}, q_4^{\prime}) \in Q_{2,4}$ such that
$\delta_{2,4}((q_2, q_4), e) =  (q_2, q_4)^{\prime} = \left\{
\begin{array}{ll}
    \left(\delta_2(q_2, e), \delta_4(q_4, e)\right), & \hbox{if $
     e\in E_2 \cap E_4$};\\
    \left(\delta_2(q_2, e), q_4\right), & \hbox{if $e\in E_2 \backslash E_4$;} \\
    \left(q_2, \delta_4(q_4, e)\right), & \hbox{if $e\in E_4 \backslash E_2$.}
\end{array}\right.$

Therefore, $\forall (q_1, q_3), (q_1, q_3)^{\prime}\in Q_{1,3}$,
$(q_2, q_4) \in Q_{2,4}$, $e\in
 E$, such that $\delta_{1,3}((q_1, q_3), e) = (q_1, q_3)^{\prime}$ and $((q_1, q_3), (q_2, q_4))\in R$,
 then $\exists (q_2, q_4)^{\prime}\in
 Q_{2,4}$, $\delta_{2,4}((q_2, q_4), e) = (q_2, q_4)^{\prime}$, $((q_1, q_3)^{\prime}, (q_2, q_4)^{\prime})\in
 R$. This together with $((q_1^0, q_3^0), (q_2^0, q_4^0))\in R$, by
construction, leads to $A_1||A_3 \prec A_2||A_4$.

Now, to prove Lemma $\ref{ParallelSimulation}.2$, we define the
relation $\bar{R} = \{ ((q_2,q_4), (q_1,q_3)) | (q_2$, $q_1) \in
\bar{R}_1 \hbox{ and }\\ (q_4,q_3) \in \bar{R}_2 \}$, where,
$\bar{R}_1$ and $\bar{R}_2$ are the respective simulation relations
from $A_2$ to $A_1$ and from $A_4$ to $A_3$, and then similar to the
proof of the first part, we show that $\bar{R}$ is a simulation
relation. Now, to show that $A_1||A_3 \cong A_2||A_4$ it remains to
show that $\forall (q_1, q_3)\in Q_{1,3}$, $(q_2, q_4)\in Q_{2,4}$:
$((q_1,q_3), (q_2,q_4)) \in R \Leftrightarrow ((q_2,q_4),
(q_1,q_3))\in \bar{R}$. This is proven by contradiction. Suppose
that $\exists (q_1, q_3)\in Q_{1,3}$, $(q_2, q_4)\in Q_{2,4}$ such
that $((q_1,q_3), (q_2,q_4)) \in R \wedge ((q_2,q_4),
(q_1,q_3))\notin \bar{R}$, or $((q_2,q_4), (q_1,q_3)) \in \bar{R}
\wedge ((q_1,q_3), (q_2,q_4))\notin R$. We prove that the first
hypothesis leads to contradiction, and the contradiction of the
second hypothesis is followed, similarly. The expression
$((q_1,q_3), (q_2,q_4)) \in R\ \wedge\ ((q_2,q_4), (q_1,q_3))\notin
\bar{R}$ means that $\exists s\in E^*$, $\delta_{1,3}((q_1^0,
q_3^0), s) = (q_1,q_3)$, $\delta_{2,4}((q_2^0, q_4^0), s) =
(q_2,q_4)$, $\forall e\in E$, $\delta_{1,3}((q_1, q_3), e)!$:
$\delta_{2,4}((q_2, q_4), e)!$; but, $\exists \sigma \in E$,
$\delta_{2,4}((q_2, q_4), \sigma)! \wedge \neg \delta_{1,3}((q_1,
q_3), \sigma)!$. From Definition \ref{parallel composition},
$\delta_{2,4}((q_2, q_4), \sigma)!$ means that\\ $\left\{
\begin{array}{ll}
   \delta _2(q_2,\sigma)!, \delta _4(q_4, \sigma)! & \hbox{if $e\in E_2 \cap E_4$};\\
   \delta _2(q_2,\sigma)! & \hbox{if $e\in E_2 \backslash E_4$};\\
   \delta _4(q_4, \sigma)! & \hbox{if $e\in E_4 \backslash E_2$}.
\end{array}\right.$\\
Consequently, from $(q_2, q_1)\in \bar{R}_1$ and $E_1 = E_2$ (due to
$A_1 \cong A_2$), $(q_4, q_3)\in \bar{R}_2$ and $E_3 = E_4$ (due to
$A_3 \cong A_4$), and Definition \ref{simulation}, it follows  that
$\left\{
\begin{array}{ll}
    \delta _1(q_1,\sigma)!, \delta _3(q_3, \sigma)! & \hbox{if $e\in E_2 \cap E_4 = E_1 \cap E_3$};\\
    \delta _1(q_1,\sigma)! & \hbox{if $e\in E_2 \backslash E_4 = E_1 \backslash E_3$};\\
    \delta _3(q_3, \sigma)! & \hbox{if $e\in E_4 \backslash E_2 = E_3 \backslash E_1$},
\end{array}\right.$ that from Definition \ref{parallel composition}
leads to $\delta_{1, 3}((q_1, q_3), \sigma)!$ which contradicts with
the hypothesis and the proof is followed.

\bibliographystyle{IEEEtran}
\bibliography{automatica2010_3}
\end{document}